\documentclass[peerreview,a4paper,12pt]{IEEEtran}

\usepackage{amsthm}
\usepackage{amsmath}
\usepackage{dsfont}
\usepackage{amsfonts,amssymb,amsmath}
\usepackage{graphicx}
\usepackage{epsfig}
\usepackage[applemac]{inputenc}
\usepackage{latexsym}

\newtheorem{mycor}{Corollary}
\newtheorem{myprop}{Proposition}
\newtheorem{mythe}{Theorem}
\newtheorem{mylem}{Lemma}

\DeclareMathOperator*{\DMC}{DMC}
\DeclareMathOperator*{\Proj}{Proj}

\DeclareMathOperator*{\supp}{supp}

\DeclareMathOperator*{\Imag}{Im}

\DeclareMathOperator*{\rank}{rank}

\DeclareMathOperator*{\Mul}{Mul}
\DeclareMathOperator*{\MP}{MP}

\begin{document}

\sloppy

\title{Continuity of Channel Parameters and Operations under Various DMC Topologies}


%
%

\author{Rajai Nasser,\\
EPFL, Lausanne, Switzerland\\
rajai.nasser@epfl.ch
}




\maketitle

\begin{abstract}
We study the continuity of many channel parameters and operations under various topologies on the space of equivalent discrete memoryless channels (DMC). We show that mutual information, channel capacity, Bhattacharyya parameter, probability of error of a fixed code, and optimal probability of error for a given code rate and blocklength, are continuous under various DMC topologies. We also show that channel operations such as sums, products, interpolations, and Ar{\i}kan-style transformations are continuous.
\end{abstract}

\section{Introduction}

Let $\mathcal{X}$ and $\mathcal{Y}$ be two finite sets and let $W$ be a fixed channel with input alphabet $\mathcal{X}$ and output alphabet $\mathcal{Y}$. It is well known that the input-output mutual information is continuous on the simplex of input probability distributions. Many other parameters that depend on the input probability distribution were shown to be continuous on the simplex in \cite{FiniteBlock}.

Polyanskiy studied in \cite{Saddle} the continuity of the Neyman-Pearson function for a binary hypothesis test that arises in the analysis of channel codes. He showed that for arbitrary input and output alphabets, this function is continuous in the input distribution in the total variation topology. He also showed that under some regularity assumptions, this function is continuous in the weak-$\ast$ topology.

If $\mathcal{X}$ and $\mathcal{Y}$ are finite sets, the space of channels with input alphabet $\mathcal{X}$ and output alphabet $\mathcal{Y}$ can naturally be endowed with the topology of the Euclidean metric, or any other equivalent metric. It is well known that the channel capacity is continuous in this topology. If $\mathcal{X}$ and $\mathcal{Y}$ are arbitrary, one can construct a topology on the space of channels using the weak-$\ast$ topology on the output alphabet. It was shown in \cite{Schwarte} that the capacity is lower semi-continuous in this topology.

The continuity results that are mentioned in the previous paragraph do not take into account ``equivalence" between channels. Two channels are said to be equivalent if they are degraded from each other. This means that each channel can be simulated from the other by local operations at the receiver. Two channels that are degraded from each other are completely equivalent from an operational point of view: both channels have exactly the same probability of error under optimal decoding for any fixed code. Moreover, any sub-optimal decoder for one channel can be transformed to a sub-optimal decoder for the other channel with the same probability of error and essentially the same computational complexity. This is why it makes sense, from an information-theoretic point of view, to identify equivalent channels and consider them as one point in the space of ``equivalent channels".

In \cite{RichardsonUrbanke}, equivalent binary-input channels were identified with their $L$-density (i.e., the density of log-likelihood ratios). The space of equivalent binary-input channels was endowed with the topology of convergence in distribution of $L$-densities. Since the symmetric capacity\footnote{The symmetric capacity is the input-output mutual information with uniformly distributed input.} and the Bhattacharyya parameter can be written as an integral of a continuous function with respect to the $L$-density \cite{RichardsonUrbanke}, it immediately follows that these parameters are continuous in the $L$-density topology.

In \cite{RajDMCTop}, many topologies were constructed for the space of equivalent channels sharing a fixed input alphabet. In this paper, we study the continuity of many channel parameters and operations under these topologies.

In Section II, we introduce the preliminaries for this paper. In Section III, we recall the main results of \cite{RajDMCTop} that we need here. In Section IV, we introduce the channel parameters and operations that we investigate in this paper. In Section V, we study the continuity of these parameters and operations in the quotient topology of the space of equivalent channels with fixed input and output alphabets. The continuity in the strong topology of the space of equivalent channels sharing the same input alphabet is studied in Section VI. Finally, the continuity in the noisiness/weak-$\ast$ and the total variation topologies is studied in Section VII.

\section{Preliminaries}

We assume that the reader is familiar with the basic concepts of general topology. The main concepts and theorems that we need can be found in the preliminaries section of \cite{RajDMCTop}.

\subsection{Set-theoretic notations}

For every integer $n\geq 1$, we denote the set $\{1,\ldots,n\}$ as $[n]$.

The set of mappings from a set $A$ to a set $B$ is denoted as $B^A$.

Let $A$ be a subset of $B$. The \emph{indicator mapping} $\mathds{1}_{A,B}:B\rightarrow\{0,1\}$ of $A$ in $B$ is defined as:
$$\mathds{1}_{A,B}(x)=\mathds{1}_{x\in A}=\begin{cases}1\quad&\text{if}\;x\in A,\\0\quad&\text{otherwise}.\end{cases}$$
If the superset $B$ is clear from the context, we simply write $\mathds{1}_A$ to denote the indicator mapping of $A$ in $B$.

The \emph{power set} of $B$ is the set of subsets of $B$. Since every subset of $B$ can be identified with its indicator mapping, we denote the power set of $B$ as $\{0,1\}^B=2^B$.

Let $(A_i)_{i\in I}$ be a collection of arbitrary sets indexed by $I$. The \emph{disjoint union} of $(A_i)_{i\in I}$ is defined as $\displaystyle \coprod_{i\in I} A_i=\bigcup_{i\in I}(A_i\times\{i\})$. For every $i\in I$, the $i^{th}$-\emph{canonical injection} is the mapping $\phi_i:A_i\rightarrow \displaystyle\coprod_{j\in I} A_j$ defined as $\phi_i(x_i)=(x_i,i)$. If no confusions can arise, we can identify $A_i$ with $A_i\times\{i\}$ through the canonical injection. Therefore, we can see $A_i$ as a subset of $\displaystyle\coprod_{j\in I} A_j$ for every $i\in I$.

Let $R$ be an equivalence relation on a set $T$. For every $x\in T$, the set $\hat{x}=\{y\in T:\; x R y\}$ is the \emph{$R$-equivalence class} of $x$. The collection of $R$-equivalence classes, which we denote as $T/R$, forms a partition of $T$, and it is called the \emph{quotient space of $T$ by $R$}. The mapping $\Proj_R:T\rightarrow T/R$ defined as $\Proj_R(x)=\hat{x}$ for every $x\in T$ is the \emph{projection mapping onto $T/R$}.

\subsection{Topological notations}

A topological space $(T,\mathcal{U})$ is said to be \emph{contractible} to $x_0\in T$ if there exists a continuous mapping $H:T\times[0,1]\to T$ such that $H(x,0)=x$ and $H(x,1)=x_0$ for every $x\in T$, where $[0,1]$ is endowed with the Euclidean topology. $(T,\mathcal{U})$ is \emph{strongly contractible} to $x_0\in T$ if we also have $H(x_0,t)=x_0$ for every $t\in[0,1]$.

Intuitively, $T$ is contractible if it can be ``continuously shrinked" to a single point $x_0$. If this ``continuous shrinking" can be done without moving $x_0$, $T$ is strongly contractible.

Note that contractibility is a very strong notion of connectedness: every contractible space is path-connected and simply connected. Moreover, all its homotopy, homology and cohomology groups of order $\geq 1$ are zero.

Let $\{(T_i,\mathcal{U}_i)\}_{i\in I}$ be a collection of topological spaces indexed by $I$. The \emph{product topology} on $\displaystyle\prod_{i\in I} T_i$ is denoted by $\displaystyle \bigotimes_{i\in I}\mathcal{U}_i$. The \emph{disjoint union topology} on $\displaystyle\coprod_{i\in I} T_i$ is denoted by $\displaystyle \bigoplus_{i\in I}\mathcal{U}_i$.

The following lemma is useful to show the continuity of many functions.
\begin{mylem}
\label{lemCompactProdCont}
Let $(S,\mathcal{V})$ and $(T,\mathcal{U})$ be two compact topological spaces and let $f:S\times T\rightarrow\mathbb{R}$ be a continuous function on $S\times T$. For every $s\in\mathcal{S}$ and every $\epsilon>0$, there exists a neighborhood $V_s$ of $s$ such that for every $s'\in V_s$, we have
$$\sup_{t\in T}|f(s',t)-f(s,t)|\leq \epsilon.$$
\end{mylem}
\begin{proof}
See Appendix \ref{appCompactProdCont}.
\end{proof}

\subsection{Quotient topology}

\label{subsecQuotient}
Let $(T,\mathcal{U})$ be a topological space and let $R$ be an equivalence relation on $T$. The \emph{quotient topology} on $T/R$ is the finest topology that makes the projection mapping $\Proj_R$ continuous. It is given by
$$\mathcal{U}/R=\left\{\hat{U}\subset T/R:\;\textstyle\Proj_R^{-1}(\hat{U})\in \mathcal{U}\right\}.$$

\begin{mylem}
\label{lemQuotientFunction}
Let $f:T\rightarrow S$ be a continuous mapping from $(T,\mathcal{U})$ to $(S,\mathcal{V})$. If $f(x)=f(x')$ for every $x,x'\in T$ satisfying $x R x'$, then we can define a \emph{transcendent mapping} $f:T/R\rightarrow S$ such that $f(\hat{x})=f(x')$ for any $x'\in\hat{x}$. $f$ is well defined on $T/R$ . Moreover, $f$ is a continuous mapping from $(T/R,\mathcal{U}/R)$ to $(S,\mathcal{V})$.
\end{mylem}

Let $(T,\mathcal{U})$ and $(S,\mathcal{V})$ be two topological spaces and let $R$ be an equivalence relation on $T$. Consider the equivalence relation $R'$ on $T\times S$ defined as $(x_1,y_1) R' (x_2,y_2)$ if and only if $x_1 R x_2$ and $y_1=y_2$. A natural question to ask is whether the canonical bijection between $\big((T/R)\times S,(\mathcal{U}/R)\otimes \mathcal{V} \big)$ and $\big((T\times S)/R',(\mathcal{U}\otimes \mathcal{V})/R' \big)$ is a homeomorphism. It turns out that this is not the case in general. The following theorem, which is widely used in algebraic topology, provides a sufficient condition:

\begin{mythe}
\label{theQuotientProd}
\cite{Engelking}
If $(S,\mathcal{V})$ is locally compact and Hausdorff, then the canonical bijection between $\big((T/R)\times S,(\mathcal{U}/R)\otimes \mathcal{V} \big)$ and $\big((T\times S)/R',(\mathcal{U}\otimes \mathcal{V})/R' \big)$ is a homeomorphism.
\end{mythe}

\begin{mycor}
\label{corQuotientProd}
Let $(T,\mathcal{U})$ and $(S,\mathcal{V})$ be two topological spaces, and let $R_T$ and $R_S$ be two equivalence relations on $T$ and $S$ respectively. Define the equivalence relation $R$ on $T\times S$ as $(x_1,y_1) R (x_2,y_2)$ if and only if $x_1 R_T x_2$ and $y_1R_S y_2$. If $(S,\mathcal{V})$ and $(T/R_T,\mathcal{U}/R_T)$ are locally compact and Hausdorff, then the canonical bijection between $\big((T/R_T)\times (S/R_S),(\mathcal{U}/R_T)\otimes (\mathcal{V}/R_S) \big)$ and $\big((T\times S)/R,(\mathcal{U}\otimes \mathcal{V})/R \big)$ is a homeomorphism.
\end{mycor}
\begin{proof}
We just need to apply Theorem \ref{theQuotientProd} twice. Define the equivalence relation $R_T'$ on $T\times S$ as follows: $(x_1,y_1) R_T' (x_2,y_2)$ if and only if $x_1 R_T x_2$ and $y_1=y_2$. Since $(S,\mathcal{V})$ is locally compact and Hausdorff, Theorem \ref{theQuotientProd} implies that the canonical bijection from $\big((T/R_T)\times S,(\mathcal{U}/R_T)\otimes \mathcal{V} \big)$ to $\big((T\times S)/R_T',(\mathcal{U}\otimes \mathcal{V})/R_T' \big)$ is a homeomorphism. Let us identify these two spaces through the canonical bijection.

Now define the equivalence relation $R_S'$ on $(T/R_T)\times S$ as follows: $(\hat{x}_1,y_1) R_S' (\hat{x}_2,y_2)$ if and only if $\hat{x}_1 = \hat{x}_2$ and $y_1 R_S y_2$. Since $(T/R_T,\mathcal{U}/R_T)$ is locally compact and Hausdorff, Theorem \ref{theQuotientProd} implies that the canonical bijection from $\big((T/R_T)\times (S/R_S),(\mathcal{U}/R_T)\otimes (\mathcal{V}/R_S) \big)$ to $\big(((T/R_T)\times S)/R_S',((\mathcal{U}/R_T)\otimes \mathcal{V})/R_S' \big)$ is a homeomorphism.

Since we identified $\big((T/R_T)\times S,(\mathcal{U}/R_T)\otimes \mathcal{V} \big)$ and $\big((T\times S)/R_T',(\mathcal{U}\otimes \mathcal{V})/R_T' \big)$ through the canonical bijection (which is a homeomorphism), $R_S'$ can be seen as an equivalence relation on $\big((T\times S)/R_T',(\mathcal{U}\otimes \mathcal{V})/R_T' \big)$. It is easy to see that the canonical bijection from $\big(\big((T\times S)/R_T'\big)/R_S',\big((\mathcal{U}\otimes \mathcal{V})/R_T'\big)/R_S' \big)$ to $\big((T\times S)/R,(\mathcal{U}\otimes \mathcal{V})/R \big)$ is a homeomorphism. We conclude that the canonical bijection from $\big((T/R_T)\times (S/R_S),(\mathcal{U}/R_T)\otimes (\mathcal{V}/R_S) \big)$ to $\big((T\times S)/R,(\mathcal{U}\otimes \mathcal{V})/R \big)$ is a homeomorphism.
\end{proof}

\subsection{Measure-theoretic notations}

If $(M,\Sigma)$ is a measurable space, we denote the set of probability measures on $(M,\Sigma)$ as $\mathcal{P}(M,\Sigma)$. If the $\sigma$-algebra $\Sigma$ is known from the context, we simply write $\mathcal{P}(M)$ to denote the set of probability measures.

If $P\in\mathcal{P}(M,\Sigma)$ and $\{x\}$ is a measurable singleton, we simply write $P(x)$ to denote $P(\{x\})$.

For every $P_1,P_2\in\mathcal{P}(M,\Sigma)$, the \emph{total variation distance} between $P_1$ and $P_2$ is defined as:
$$\|P_1-P_2\|_{TV}=\sup_{A\in\Sigma}|P_1(A)-P_2(A)|.$$

\vspace{3mm}

\emph{The push-forward probability measure}\\
Let $P$ be a probability measure on $(M,\Sigma)$, and let $f:M\rightarrow M'$ be a measurable mapping from $(M,\Sigma)$ to another measurable space $(M',\Sigma')$. The \emph{push-forward probability measure of $P$ by $f$} is the probability measure $f_{\#}P$ on $(M',\Sigma')$ defined as $(f_{\#}P)(A')=P(f^{-1}(A'))$ for every $A'\in\Sigma'$.

A measurable mapping $g:M'\rightarrow\mathbb{R}$ is integrable with respect to $f_{\#}P$ if and only if $g\circ f$ is integrable with respect to $P$. Moreover,
$$\int_{M'}g\cdot d(f_{\#}P)=\int_{M}(g\circ f)\cdot dP.$$

The mapping $f_{\#}$ from $\mathcal{P}(M,\Sigma)$ to $\mathcal{P}(M',\Sigma')$ is continuous if these spaces are endowed with the total variation topology:
\begin{align*}
\|f_{\#}P-f_{\#}P'\|_{TV}&=\sup_{A'\in\Sigma'} |(f_{\#}P)(A') - (f_{\#}P')(A')|\\
&= \sup_{A'\in\Sigma'} |P(f^{-1}(A'))-P'(f^{-1}(A'))|\leq \sup_{A\in\Sigma} |P(A)-P'(A)|\leq \|P-P'\|_{TV}.
\end{align*}

\vspace{3mm}

\emph{Probability measures on finite sets}\\

We always endow finite sets with their finest $\sigma$-algebra, i.e., the power set. In this case, every probability measure is completely determined by its value on singletons, i.e., if $P$ is a probability measure on a finite set $\mathcal{X}$, then for every $A\subset\mathcal{X}$, we have
$$P(A)=\sum_{x\in A}P(x).$$

If $\mathcal{X}$ is a finite set, we denote the set of probability distributions on $\mathcal{X}$ as $\Delta_{\mathcal{X}}$. Note that $\Delta_{\mathcal{X}}$ is an $(|\mathcal{X}|-1)$-dimensional simplex in $\mathbb{R}^{\mathcal{X}}$. We always endow $\Delta_{\mathcal{X}}$ with the total variation distance and its induced topology. For every $p_1,p_2\in\Delta_{\mathcal{X}}$, we have:
$$\|p_1-p_2\|_{TV}=\frac{1}{2}\sum_{x\in\mathcal{X}}|p_1(x)-p_2(x)|=\frac{1}{2}\|p_1-p_2\|_1.$$

\vspace{3mm}

\emph{Products of probability measures}\\
We denote the product of two measurable spaces $(M_1,\Sigma_1)$ and $(M_2,\Sigma_2)$ as $(M_1\times M_2,\Sigma_1\otimes\Sigma_2)$. If $P_1\in\mathcal{P}(M_1,\Sigma_1)$ and $P_2\in\mathcal{P}(M_2,\Sigma_2)$, we denote the product of $P_1$ and $P_2$ as $P_1\times P_2$.

If $\mathcal{P}(M_1,\Sigma_1)$, $\mathcal{P}(M_2,\Sigma_2)$ and $\mathcal{P}(M_1\times M_2,\Sigma_1\otimes\Sigma_2)$ are endowed with the total variation topology, the mapping $(P_1,P_2)\rightarrow P_1\times P_2$ is a continuous mapping (see Appendix \ref{appProdMeasureCont}).

\vspace{3mm}

\emph{Borel sets and the support of a probability measure}\\
Let $(T,\mathcal{U})$ be a Hausdorff topological space. The Borel $\sigma$-algebra of $(T,\mathcal{U})$ is the $\sigma$-algebra generated by $\mathcal{U}$. We denote the Borel $\sigma$-algebra of $(T,\mathcal{U})$ as $\mathcal{B}(T,\mathcal{U})$. If the topology $\mathcal{U}$ is known from the context, we simply write $\mathcal{B}(T)$ to denote the Borel $\sigma$-algebra. The sets in $\mathcal{B}(T)$ are called the Borel sets  of $T$.

The \emph{support} of a measure $P\in\mathcal{P}(T,\mathcal{B}(T))$ is the set of all points $x\in T$ for which every neighborhood has a strictly positive measure:
$$\supp(P)=\{x\in T:\;P(O)>0\;\text{for every neighborhood}\;O\;\text{of}\;x\}.$$
If $P$ is a probability measure on a Polish space, then $P\big(T\setminus\supp(P)\big)=0$.

\subsection{Random mappings}

\label{subsecRandMap}

Let $M$ and $M'$ be two arbitrary sets and let $\Sigma'$ be a $\sigma$-algebra on $M'$. A \emph{random mapping} from $M$ to $(M',\Sigma')$ is a mapping $R$ from $M$ to $\mathcal{P}(M',\Sigma')$. For every $x\in M$, $R(x)$ can be interpreted as the probability distribution of the random output given that the input is $x$.

Let $\Sigma$ be a $\sigma$-algebra on $M$. We say that $R$ is a \emph{measurable random mapping} from $(M,\Sigma)$ to $(M',\Sigma')$ if the mapping $R_B:M\rightarrow\mathbb{R}$ defined as $R_B(x)=(R(x))(B)$ is measurable for every $B\in \Sigma'$.

Note that this definition of measurability is consistent with the measurability of ordinary mappings: let $f$ be a mapping from $M$ to $M'$ and let $D_f:M\rightarrow\mathcal{P}(M',\Sigma')$ be the random mapping defined as $D_f(x)=\delta_{f(x)}$ for every $x\in M$, where $\delta_{f(x)}\in \mathcal{P}(M',\Sigma')$ is a Dirac measure centered at $f(x)$. We have:
\begin{align*}
D_f\;\text{is measurable}\;\;&\Leftrightarrow\;\; (D_f)_B\;\text{is measurable},\;\;\forall B\in\Sigma'\\
&\Leftrightarrow\;\; ((D_f)_B)^{-1}(B')\in\Sigma, \;\;\forall B'\in\mathcal{B}(\mathbb{R}),\; \forall B\in\Sigma'\\
&\stackrel{(a)}{\Leftrightarrow}\;\; ((D_f)_B)^{-1}(\{1\})\in\Sigma, \;\;\forall B\in\Sigma'\\
&\stackrel{(b)}{\Leftrightarrow}\;\; f^{-1}(B)\in\Sigma, \;\;\forall B\in\Sigma'\\
&\Leftrightarrow\;\; f\;\text{is measurable},
\end{align*}
where (a) and (b) follow from the fact that $((D_f)_B)(x)$ is either 1 or 0 depending on whether $f(x)\in B$ or not.

Let $P$ be a probability measure on $(M,\Sigma)$ and let $R$ be a measurable random mapping from $(M,\Sigma)$ to $(M',\Sigma')$. The \emph{push-forward probability measure of $P$ by $R$} is the probability measure $R_{\#}P$ on $(M',\Sigma')$ defined as:
$$(R_{\#}P)(B)=\int_{M} R_B \cdot dP,\;\;\forall B\in\Sigma'.$$
Note that this definition is consistent with the push-forward of ordinary mappings: if $f$ and $D_f$ are as above, then for every $B\in\Sigma'$, we have
$$((D_f)_{\#}P)(B)=\int_{M} (D_f)_B \cdot dP=\int_{M} (\mathds{1}_B\circ f) \cdot dP=\int_{M'} \mathds{1}_B\cdot d (f_{\#}P)=(f_{\#}P)(B).$$

\begin{myprop}
\label{propPushForwardRandFormula}
Let $R$ be a measurable random mapping from $(M,\Sigma)$ to $(M',\Sigma')$. If $g:M'\rightarrow\mathbb{R}^+\cup\{+\infty\}$ is a $\Sigma'$-measurable mapping, then the mapping $x\rightarrow\displaystyle\int_{M'}g(y)\cdot d(R(x))(y)$ is a measurable mapping from $(M,\Sigma)$ to $\mathbb{R}^+\cup\{+\infty\}$. Moreover, for every $P\in\mathcal{P}(M,\Sigma)$, we have
$$\int_{M'}g\cdot d(R_{\#}P)=\int_M\left(\int_{M'}g(y)\cdot d(R(x))(y)\right)dP(x).$$
\end{myprop}
\begin{proof}
See Appendix \ref{appPushForwardRandFormula}.
\end{proof}

\begin{mycor}
\label{corPushForwardRandFormula}
If $g:M'\rightarrow\mathbb{R}$ is bounded and $\Sigma'$-measurable, then the mapping $$x\rightarrow\displaystyle\int_{M'}g(y)\cdot d(R(x))(y)$$ is bounded and $\Sigma$-measurable. Moreover, for every $P\in\mathcal{P}(M,\Sigma)$, we have
$$\int_{M'}g\cdot d(R_{\#}P)=\int_M\left(\int_{M'}g(y)\cdot d(R(x))(y)\right)dP(x).$$
\end{mycor}
\begin{proof}
Write $g=g^+-g^-$ (where $g^+=\max\{g,0\}$ and $g^-=\max\{-g,0\}$), and use the fact that every bounded measurable function is integrable over any probability distribution.
\end{proof}

\begin{mylem}
\label{lemContPushForwardRandTV}
For every measurable random mapping $R$ from $(M,\Sigma)$ to $(M',\Sigma')$, the push-forward mapping $R_{\#}$ is continuous from $\mathcal{P}(M,\Sigma)$ to $\mathcal{P}(M',\Sigma')$ under the total variation topology.
\end{mylem}
\begin{proof}
See Appendix \ref{appContPushForwardRand}.
\end{proof}

\begin{mylem}
\label{lemContPushForwardRandWeakStar}
Let $\mathcal{U}$ be a Polish\footnote{This assumption can be dropped. We assumed that $\mathcal{U}$ is Polish just to avoid working with Moore-Smith nets.} topology on $M$, and let $\mathcal{U}'$ be an arbitrary topology on $M'$. Let $R$ be a measurable random mapping from $(M,\mathcal{B}(M))$ to $(M',\mathcal{B}(M'))$. Moreover, assume that $R$ is a continuous mapping from $(M,\mathcal{U})$ to $\mathcal{P}(M',\mathcal{B}(M'))$ when the latter space is endowed with the weak-$\ast$ topology. Under these assumptions, the push-forward mapping $R_{\#}$ is continuous from $\mathcal{P}(M,\mathcal{B}(M))$ to $\mathcal{P}(M',\mathcal{B}(M'))$ under the weak-$\ast$ topology.
\end{mylem}
\begin{proof}
See Appendix \ref{appContPushForwardRand}.
\end{proof}

\subsection{Meta-probability measures}

\label{subsecMetaProbDef}

Let $\mathcal{X}$ be a finite set. A \emph{meta-probability measure} on $\mathcal{X}$ is a probability measure on the Borel sets of $\Delta_{\mathcal{X}}$. It is called a meta-probability measure because it is a probability measure on the space of probability distributions on $\mathcal{X}$.

We denote the set of meta-probability measures on $\mathcal{X}$ as $\mathcal{MP}(\mathcal{X})$. Clearly, $\mathcal{MP}(\mathcal{X})=\mathcal{P}(\Delta_{\mathcal{X}})$.

A meta-probability measure ${\MP}$ on $\mathcal{X}$ is said to be \emph{balanced} if it satisfies
$$\int_{\Delta_{\mathcal{X}}}p\cdot d{\MP}(p)=\pi_{\mathcal{X}},$$
where $\pi_{\mathcal{X}}$ is the uniform probability distribution on $\mathcal{X}$.

We denote the set of all balanced meta-probability measures on $\mathcal{X}$ as $\mathcal{MP}_b(\mathcal{X})$. The set of all balanced and finitely supported meta-probability measures on $\mathcal{X}$ is denoted as $\mathcal{MP}_{bf}(\mathcal{X})$.

The following lemma is useful to show the continuity of functions defined on $\mathcal{MP}(\mathcal{X})$.
\begin{mylem}
\label{lemContProdWeakStar}
Let $(S,\mathcal{V})$ be a compact topological space and let $f:S\times\Delta_{\mathcal{X}}\rightarrow\mathbb{R}$ be a continuous function on $S\times \Delta_{\mathcal{X}}$. The mapping $F:S\times\mathcal{MP}(\mathcal{X})\rightarrow\mathbb{R}$ defined as
$$F(s,{\MP})=\int_{\Delta_{\mathcal{X}}} f(s,p)\cdot d{\MP}(p)$$
is continuous, where $\mathcal{MP}(\mathcal{X})$ is endowed with the weak-$\ast$ topology.
\end{mylem}
\begin{proof}
See Appendix \ref{appContProdWeakStar}.
\end{proof}

\vspace*{3mm}

Let $f$ be a mapping from a finite set $\mathcal{X}$ to another finite set $\mathcal{X}'$. $f$ induces a push-forward mapping $f_{\#}$ taking probability distributions in $\Delta_{\mathcal{X}}$ to probability distributions in $\Delta_{\mathcal{X}'}$. $f_{\#}$ is continuous because $\Delta_{\mathcal{X}}$ and $\Delta_{\mathcal{X}'}$ are endowed with the total variation distance. $f_{\#}$ in turn induces another push-forward mapping taking meta-probability measures in $\mathcal{MP}(\mathcal{X})$ to meta-probability measures in $\mathcal{MP}(\mathcal{X}')$. We denote this mapping as $f_{\#\#}$ and we call it \emph{the meta-push-forward mapping} induced by $f$. Since $f_{\#}$ is a continuous mapping from $\Delta_{\mathcal{X}}$ to $\Delta_{\mathcal{X}'}$, $f_{\#\#}$ is a continuous mapping from $\mathcal{MP}(\mathcal{X})$ to $\mathcal{MP}(\mathcal{X}')$ under both the weak-$\ast$ and the total variation topologies.

Let $\mathcal{X}_1$ and $\mathcal{X}_2$ be two finite sets. Let $\Mul:\Delta_{\mathcal{X}_1}\times \Delta_{\mathcal{X}_2}\rightarrow\Delta_{\mathcal{X}_1\times\mathcal{X}_2}$ be defined as $\Mul(p_1,p_2)=p_1\times p_2$. For every ${\MP}_1\in\mathcal{MP}(\mathcal{X}_1)$ and ${\MP}_2\in\mathcal{MP}(\mathcal{X}_2)$, we define the \emph{tensor product} of $\MP_1$ and $\MP_2$ as $\MP_1\otimes\MP_2=\Mul_{\#}(\MP_1\times\MP_2)\in\mathcal{MP}(\mathcal{X}_1\times\mathcal{X}_2)$.

Note that since $\Delta_{\mathcal{X}_1}$, $\Delta_{\mathcal{X}_2}$ and $\Delta_{\mathcal{X}_1\times\mathcal{X}_2}$ are endowed with the total variation topology, $\Mul(p_1,p_2)=p_1\times p_2$ is a continuous mapping from $\Delta_{\mathcal{X}_1}\times\Delta_{\mathcal{X}_2}$ to $\Delta_{\mathcal{X}_1\times\mathcal{X}_2}$. Therefore, $\Mul_{\#}$ is a continuous mapping from $\mathcal{P}(\Delta_{\mathcal{X}_1}\times\Delta_{\mathcal{X}_2})$ to $\mathcal{P}(\Delta_{\mathcal{X}_1\times\mathcal{X}_2})=\mathcal{MP}(\mathcal{X}_1\times\mathcal{X}_2)$ under both the weak-$\ast$ and the total variation topologies. On the other hand, Appendices \ref{appProdMeasureCont} and \ref{appContMetaProbProd} imply that the mapping $({\MP}_1,{\MP}_2)\rightarrow{\MP}_1\times{\MP}_2$ from $\mathcal{MP}(\mathcal{X}_1)\times\mathcal{MP}(\mathcal{X}_2)$ to $\mathcal{P}(\Delta_{\mathcal{X}_1}\times \Delta_{\mathcal{X}_2})$ is continuous under both the weak-$\ast$ and the total variation topologies. We conclude that the tensor product is continuous under both these topologies.

\section{The space of equivalent channels}

In this section, we summarize the main results of \cite{RajDMCTop}.

\subsection{Space of channels from $\mathcal{X}$ to $\mathcal{Y}$}

A discrete memoryless channel $W$ is a 3-tuple $W=(\mathcal{X},\;\mathcal{Y},\;p_W)$ where $\mathcal{X}$ is a finite set that is called the \emph{input alphabet} of $W$, $\mathcal{Y}$ is a finite set that is called the \emph{output alphabet} of $W$, and $p_W:\mathcal{X} \times \mathcal{Y} \rightarrow [0,1]$  is a function satisfying $\forall x\in\mathcal{X},\;\displaystyle\sum_{y\in\mathcal{Y}} p_W(x,y)=1$.

For every $(x,y)\in\mathcal{X}\times\mathcal{Y}$, we denote $p_W(x,y)$ as $W(y|x)$, which we interpret as the conditional probability of receiving $y$ at the output, given that $x$ is the input.

Let $\DMC_{\mathcal{X},\mathcal{Y}}$ be the set of all channels having $\mathcal{X}$ as input alphabet and $\mathcal{Y}$ as output alphabet.

For every $W,W'\in\DMC_{\mathcal{X},\mathcal{Y}}$, define the distance between $W$ and $W'$ as follows:
$$d_{\mathcal{X},\mathcal{Y}}(W,W')=\frac{1}{2} \max_{x\in\mathcal{X}}\sum_{y\in\mathcal{Y}}|W'(y|x)-W(y|x)|.$$
We always endow $\DMC_{\mathcal{X},\mathcal{Y}}$ with the metric distance $d_{\mathcal{X},\mathcal{Y}}$. This metric makes $\DMC_{\mathcal{X},\mathcal{Y}}$ a compact path-connected metric space. The metric topology on $\DMC_{\mathcal{X},\mathcal{Y}}$ that is induced by $d_{\mathcal{X},\mathcal{Y}}$ is denoted as $\mathcal{T}_{\mathcal{X},\mathcal{Y}}$.

\subsection{Equivalence between channels}

\label{subsecMetaProb}

Let $W\in \DMC_{\mathcal{X},\mathcal{Y}}$ and $W'\in \DMC_{\mathcal{X},\mathcal{Z}}$ be two channels having the same input alphabet. We say that $W'$ is \emph{degraded} from $W$ if there exists a channel $V\in \DMC_{\mathcal{Y},\mathcal{Z}}$ such that $$W'(z|x)=\sum_{y\in\mathcal{Y}} V(z|y) W(y|x).$$
$W$ and $W'$ are said to be \emph{equivalent} if each one is degraded from the other.

Let $\Delta_{\mathcal{X}}$ and $\Delta_{\mathcal{Y}}$ be the space of probability distributions on $\mathcal{X}$ and $\mathcal{Y}$ respectively. Define $P_W^o\in\Delta_{\mathcal{Y}}$ as $\displaystyle P_W^o(y)=\frac{1}{|\mathcal{X}|}\sum_{x\in\mathcal{X}}W(y|x)$ for every $y\in\mathcal{Y}$. The \emph{image} of $W$ is the set of output-symbols $y\in\mathcal{Y}$ having strictly positive probabilities:
$$\Imag(W)=\{y\in\mathcal{Y}:\; P_W^o(y)>0\}.$$

For every $y\in\Imag(W)$, define $W^{-1}_y\in\Delta_{\mathcal{X}}$ as follows:
$$W^{-1}_y(x)=\displaystyle\frac{W(y|x)}{|\mathcal{X}|P_W^o(y)},\;\;\forall x\in\mathcal{X}.$$

For every $(x,y)\in\mathcal{X}\times\Imag(W)$, we have $W(y|x)=|\mathcal{X}|P_W^o(y)W_y^{-1}(x)$. On the other hand, if $x\in\mathcal{X}$ and $y\in\mathcal{Y}\setminus\Imag(W)$, we have $W(y|x)=0$. This shows that $P_W^o$ and the collection $\{W_y^{-1}\}_{y\in \Imag(W)}$ uniquely determine $W$.

The \emph{Blackwell measure}\footnote{In an earlier version of this work, I called $\MP_W$ the \emph{posterior meta-probability distribution} of $W$. Maxim Raginsky thankfully brought to my attention the fact that $\MP_W$ is called \emph{Blackwell measure}.} (denoted ${\MP}_W$) of $W$ is a meta-probability measure on $\mathcal{X}$ defined as:
$${\textstyle{\MP}_W}=\sum_{y\in\Imag(W)}P_W^o(y) \delta_{W^{-1}_y},$$
where $\delta_{W^{-1}_y}$ is a Dirac measure centered at $W^{-1}_y$.

It is known that a meta-probability measure ${\MP}$ on $\mathcal{X}$ is the Blackwell measure of some DMC with input alphabet $\mathcal{X}$ if and only if it is balanced and finitely supported \cite{torgersen}.

It is also known that two channels $W\in\DMC_{\mathcal{X},\mathcal{Y}}$ and $W'\in\DMC_{\mathcal{X},\mathcal{Z}}$ are equivalent if and only if ${\MP}_W={\MP}_{W'}$ \cite{torgersen}.

\subsection{Space of equivalent channels from $\mathcal{X}$ to $\mathcal{Y}$}

\label{subsecDMCXYo}

Let $\mathcal{X}$ and $\mathcal{Y}$ be two finite sets. Define the equivalence relation $R_{\mathcal{X},\mathcal{Y}}^{(o)}$ on $\DMC_{\mathcal{X},\mathcal{Y}}$ as follows:
$$\forall W,W'\in \textstyle\DMC_{\mathcal{X},\mathcal{Y}},\;\; WR_{\mathcal{X},\mathcal{Y}}^{(o)}W'\;\;\Leftrightarrow\;\;W\;\text{is equivalent to}\;W'.$$

\emph{The space of equivalent channels with input alphabet $\mathcal{X}$ and output alphabet $\mathcal{Y}$} is the quotient of $\DMC_{\mathcal{X},\mathcal{Y}}$ by the equivalence relation:
$$\textstyle\DMC_{\mathcal{X},\mathcal{Y}}^{(o)}=\DMC_{\mathcal{X},\mathcal{Y}}/R_{\mathcal{X},\mathcal{Y}}^{(o)}.$$

\vspace*{3mm}

\emph{Quotient topology}\\
We define the topology $\mathcal{T}_{\mathcal{X},\mathcal{Y}}^{(o)}$ on $\DMC_{\mathcal{X},\mathcal{Y}}^{(o)}$ as the quotient topology $\mathcal{T}_{\mathcal{X},\mathcal{Y}}/R_{\mathcal{X},\mathcal{Y}}^{(o)}$. We always associate $\DMC_{\mathcal{X},\mathcal{Y}}^{(o)}$ with the quotient topology $\mathcal{T}_{\mathcal{X},\mathcal{Y}}^{(o)}$.

We have shown in \cite{RajDMCTop} that $\DMC_{\mathcal{X},\mathcal{Y}}^{(o)}$ is a compact, path-connected and metrizable space.

If $\mathcal{Y}_1$ and $\mathcal{Y}_2$ are two finite sets of the same size, there exists a canonical homeomorphism between $\DMC_{\mathcal{X},\mathcal{Y}_1}^{(o)}$ and $\DMC_{\mathcal{X},\mathcal{Y}_2}^{(o)}$ \cite{RajDMCTop}. This allows us to identify $\DMC_{\mathcal{X},\mathcal{Y}}^{(o)}$ with $\DMC_{\mathcal{X},[n]}^{(o)}$, where $n=|\mathcal{Y}|$ and $[n]=\{1,\ldots,n\}$.

Moreover, for every $1\leq n\leq m$, there exists a canonical subspace of $\DMC_{\mathcal{X},[m]}^{(o)}$ that is homeomorphic to $\DMC_{\mathcal{X},[n]}^{(o)}$ \cite{RajDMCTop}. Therefore, we can consider $\DMC_{\mathcal{X},[n]}^{(o)}$ as a compact subspace of $\DMC_{\mathcal{X},[m]}^{(o)}$.

\vspace*{3mm}

\emph{Noisiness metric}\\
For every $m\geq 1$, let $\Delta_{[m]\times\mathcal{X}}$ be the space of probability distributions on $[m]\times\mathcal{X}$. Let $\mathcal{Y}$ be a finite set and let $W\in\DMC_{\mathcal{X},\mathcal{Y}}$. For every $p\in \Delta_{[m]\times\mathcal{X}}$, define $P_c(p,W)$ as follows:
\begin{equation}
\label{eqProbCorrectGuess}
P_c(p,W)=\sup_{D\in \DMC_{\mathcal{Y},[m]}}\sum_{\substack{u\in[m],\\x\in\mathcal{X},\\y\in\mathcal{Y}}}p(u,x)W(y|x)D(u|y).
\end{equation}

The quantity $P_c(p,W)$ depends only on the $R_{\mathcal{X},\mathcal{Y}}^{(o)}$-equivalence class of $W$ (see \cite{RajDMCTop}). Therefore, if $\hat{W}\in\DMC_{\mathcal{X},\mathcal{Y}}^{(o)}$, we can define $P_c(p,\hat{W}):=P_c(p,W')$ for any $W'\in \hat{W}$.

Define the \emph{noisiness distance} $d_{\mathcal{X},\mathcal{Y}}^{(o)}: \DMC_{\mathcal{X},\mathcal{Y}}^{(o)}\times \DMC_{\mathcal{X},\mathcal{Y}}^{(o)}\rightarrow\mathbb{R}^+$ as follows:
$$d_{\mathcal{X},\mathcal{Y}}^{(o)}(\hat{W}_1,\hat{W}_2)=\sup_{\substack{m\geq 1,\\p\in\Delta_{[m]\times\mathcal{X}}}}|P_c(p,\hat{W}_1)-P_c(p,\hat{W}_2)|.$$

We have shown in \cite{RajDMCTop} that $(\DMC_{\mathcal{X},\mathcal{Y}}^{(o)},\mathcal{T}_{\mathcal{X},\mathcal{Y}}^{(o)})$ is topologically equivalent to $(\DMC_{\mathcal{X},\mathcal{Y}}^{(o)},d_{\mathcal{X},\mathcal{Y}}^{(o)})$.

\subsection{Space of equivalent channels with input alphabet $\mathcal{X}$}

The \emph{space of channels with input alphabet $\mathcal{X}$} is defined as $$\textstyle\DMC_{\mathcal{X},\ast}={\displaystyle\coprod_{n\geq 1}} \DMC_{\mathcal{X},[n]}.$$

We define the equivalence relation $R_{\mathcal{X},\ast}^{(o)}$ on $\DMC_{\mathcal{X},\ast}$ as follows:
$$\forall W,W'\in \textstyle\DMC_{\mathcal{X},\ast},\;\; WR_{\mathcal{X},\ast}^{(o)}W'\;\;\Leftrightarrow\;\;W\;\text{is equivalent to}\;W'.$$

The \emph{space of equivalent channels with input alphabet $\mathcal{X}$} is the quotient of $\DMC_{\mathcal{X},\ast}$ by the equivalence relation:
$$\textstyle\DMC_{\mathcal{X},\ast}^{(o)}=\DMC_{\mathcal{X},\ast}/R_{\mathcal{X},\ast}^{(o)}.$$

For every $n\geq 1$ and every $W\in\DMC_{\mathcal{X},[n]}$, we identify the $R_{\mathcal{X},[n]}^{(o)}$-equivalence class of $W$ with the $R_{\mathcal{X},\ast}^{(o)}$-equivalence class of it. This allows us to consider $\DMC_{\mathcal{X},[n]}^{(o)}$ as a subspace of $\DMC_{\mathcal{X},\ast}^{(o)}$. Moreover,
$${\DMC}_{\mathcal{X},\ast}^{(o)}=\bigcup_{n\geq 1}{\DMC}_{\mathcal{X},[n]}^{(o)}.$$

Since any two equivalent channels have the same Blackwell measure, we can define the \emph{Blackwell measure} of $\hat{W}\in\DMC_{\mathcal{X},\ast}^{(o)}$ as $\MP_{\hat{W}}=\MP_{W'}$ for any $W'\in\hat{W}$. The \emph{rank} of $\hat{W}\in\DMC_{\mathcal{X},\ast}^{(o)}$ is the size of the support of its Blackwell measure:
$$\rank(\hat{W})=|\supp({\MP}_{\hat{W}})|.$$
We have:
$${\DMC}_{\mathcal{X},[n]}^{(o)}=\{\hat{W}\in{\DMC}_{\mathcal{X},\ast}^{(o)}:\;\rank(\hat{W})\leq n\}.$$

A topology $\mathcal{T}$ on ${\DMC}_{\mathcal{X},\ast}^{(o)}$ is said to be \emph{natural} if and only if it induces the quotient topology $\mathcal{T}_{\mathcal{X},[n]}^{(o)}$ on ${\DMC}_{\mathcal{X},[n]}^{(o)}$ for every $n\geq 1$.

Every natural topology is $\sigma$-compact, separable and path-connected \cite{RajDMCTop}. On the other hand, if $|\mathcal{X}|\geq 2$, a Hausdorff natural topology is not Baire and it is not locally compact anywhere \cite{RajDMCTop}. This implies that no natural topology can be completely metrized if $|\mathcal{X}|\geq 2$.

\vspace*{3mm}

\emph{Strong topology on $\DMC_{\mathcal{X},\ast}^{(o)}$}\\
We associate $\DMC_{\mathcal{X},\ast}$ with the disjoint union topology $\mathcal{T}_{s,\mathcal{X},\ast}:=\displaystyle\bigoplus_{n\geq 1}\mathcal{T}_{\mathcal{X},[n]}$. The space $(\DMC_{\mathcal{X},\ast},\mathcal{T}_{s,\mathcal{X},\ast})$ is disconnected, metrizable and $\sigma$-compact \cite{RajDMCTop}.

The \emph{strong topology} $\mathcal{T}_{s,\mathcal{X},\ast}^{(o)}$ on $\DMC_{\mathcal{X},\ast}^{(o)}$ is the quotient of $\mathcal{T}_{s,\mathcal{X},\ast}$ by $R_{\mathcal{X},\ast}^{(o)}$:
$$\mathcal{T}_{s,\mathcal{X},\ast}^{(o)}=\mathcal{T}_{s,\mathcal{X},\ast}/R_{\mathcal{X},\ast}^{(o)}.$$
We call open and closed sets in $(\DMC_{\mathcal{X},\ast}^{(o)},\mathcal{T}_{s,\mathcal{X},\ast}^{(o)})$ as strongly open and strongly closed sets respectively. If $A$ is a subset of $\DMC_{\mathcal{X},\ast}^{(o)}$, then $A$ is strongly open if and only if $A\cap\DMC_{\mathcal{X},[n]}^{(o)}$ is open in $\DMC_{\mathcal{X},[n]}^{(o)}$ for every $n\geq 1$. Similarly, $A$ is strongly closed if and only if $A\cap\DMC_{\mathcal{X},[n]}^{(o)}$ is closed in $\DMC_{\mathcal{X},[n]}^{(o)}$ for every $n\geq 1$.

We have shown in \cite{RajDMCTop} that $\mathcal{T}_{s,\mathcal{X},\ast}^{(o)}$ is the finest natural topology. The strong topology is sequential, compactly generated, and $T_4$ \cite{RajDMCTop}. On the other hand, if $|\mathcal{X}|\geq 2$, the strong topology is not first-countable anywhere \cite{RajDMCTop}, hence it is not metrizable.

\vspace*{3mm}

\emph{Noisiness metric}\\
Define the \emph{noisiness metric} on $\DMC_{\mathcal{X},\ast}^{(o)}$ as follows: $$d_{\mathcal{X},\ast}^{(o)}(\hat{W},\hat{W}'):=d_{\mathcal{X},[n]}^{(o)}(\hat{W},\hat{W}')\;\text{where}\;n\geq 1\;\text{satisfies}\;\hat{W},\hat{W}'\in\textstyle\DMC_{\mathcal{X},[n]}^{(o)}.$$
$d_{\mathcal{X},\ast}^{(o)}(\hat{W},\hat{W}')$ is well defined because $d_{\mathcal{X},[n]}^{(o)}(\hat{W},\hat{W}')$ does not depend on $n\geq 1$ as long as $\hat{W},\hat{W}'\in\textstyle\DMC_{\mathcal{X},[n]}^{(o)}$. We can also express $d_{\mathcal{X},\ast}^{(o)}$ as follows:
$$d_{\mathcal{X},\ast}^{(o)}(\hat{W},\hat{W}')=\sup_{\substack{m\geq 1,\\p\in\Delta_{[m]\times\mathcal{X}}}}|P_c(p,\hat{W})-P_c(p,\hat{W}')|.$$

The metric topology on $\DMC_{\mathcal{X},\ast}^{(o)}$ that is induced by $d_{\mathcal{X},\ast}^{(o)}$ is called the \emph{noisiness topology} on $\DMC_{\mathcal{X},\ast}^{(o)}$, and it is denoted as $\mathcal{T}_{\mathcal{X},\ast}^{(o)}$. We have shown in \cite{RajDMCTop} that $\mathcal{T}_{\mathcal{X},\ast}^{(o)}$ is a natural topology that is strictly coarser than $\mathcal{T}_{s,\mathcal{X},\ast}^{(o)}$.

\vspace*{3mm}

\emph{Topologies from Blackwell measures}\\
The mapping $\hat{W}\rightarrow{\MP}_{\hat{W}}$ is a bijection from $\DMC_{\mathcal{X},\ast}^{(o)}$ to $\mathcal{MP}_{bf}(\mathcal{X})$. We call this mapping the \emph{canonical bijection} from $\DMC_{\mathcal{X},\ast}^{(o)}$ to $\mathcal{MP}_{bf}(\mathcal{X})$.

Since $\Delta_{\mathcal{X}}$ is a metric space, there are many standard ways to construct topologies on $\mathcal{MP}(\mathcal{X})$. If we choose any of these standard topologies on $\mathcal{MP}(\mathcal{X})$ and then relativize it to the subspace $\mathcal{MP}_{bf}(\mathcal{X})$, we can construct topologies on $\DMC_{\mathcal{X},\ast}^{(o)}$ through the canonical bijection.

In \cite{RajDMCTop}, we studied the weak-$\ast$ and the total variation topologies. We showed that the weak-$\ast$ topology is exactly the same as the noisiness topology.

The \emph{total-variation metric distance} $d_{TV,\mathcal{X},\ast}^{(o)}$ on $\DMC_{\mathcal{X},\ast}^{(o)}$ is defined as
$$d_{TV,\mathcal{X},\ast}^{(o)}(\hat{W},\hat{W}')=\|{\MP}_{\hat{W}}-{\MP}_{\hat{W}'}\|_{TV}.$$

The \emph{total-variation topology} $\mathcal{T}_{TV,\mathcal{X},\ast}^{(o)}$ is the metric topology that is induced by $d_{TV,\mathcal{X},\ast}^{(o)}$ on $\DMC_{\mathcal{X},\ast}^{(o)}$. We proved in \cite{RajDMCTop} that if $|\mathcal{X}|\geq 2$, we have:
\begin{itemize}
\item $\mathcal{T}_{TV,\mathcal{X},\ast}^{(o)}$ is not natural nor Baire, hence it is not completely metrizable.
\item $\mathcal{T}_{TV,\mathcal{X},\ast}^{(o)}$ is not locally compact anywhere.
\end{itemize}

\section{Channel parameters and operations}

\subsection{Useful parameters}

\label{subsecChanParam}

Let $\Delta_{\mathcal{X}}$ be the space of probability distributions on $\mathcal{X}$. For every $p\in\Delta_{\mathcal{X}}$ and every $W\in\DMC_{\mathcal{X},\mathcal{Y}}$, define $I(p,W)$ as the mutual information $I(X;Y)$, where $X$ is distributed as $p$ and $Y$ is the output of $W$ when $X$ is the input. The mutual information is computed using the natural logarithm. The \emph{capacity} of $W$ is defined as $\displaystyle C(W)=\sup_{p\in\Delta_{\mathcal{X}}}I(p,W)$.

For every $p\in\Delta_{\mathcal{X}}$, \emph{the error probability of the MAP decoder of $W$ under prior $p$} is defined as:
$$P_e(p,W)=1- \sum_{y\in\mathcal{Y}} \max_{x\in\mathcal{X}}\{ p(x)W(y|x)\}.$$
Clearly, $0\leq P_e(p,W)\leq 1$.

For every $W\in\DMC_{\mathcal{X},\mathcal{Y}}$, define the \emph{Bhattacharyya parameter} of $W$  as:
$$Z(W)=\begin{cases}\displaystyle\frac{1}{|\mathcal{X}|\cdot(|\mathcal{X}|-1)}\sum_{\substack{x_1,x_2\in\mathcal{X},\\x_1\neq x_2}}\sum_{y\in\mathcal{Y}}\sqrt{W(y|x_1)W(y|x_2)},\quad&\text{if}\;|\mathcal{X}|\geq 2\\
0\quad&\text{if}\;|\mathcal{X}|=1.\end{cases}$$
It is easy to see that $0\leq Z(W)\leq 1$.

It was shown in \cite{SasogluTelAri} and \cite{RajErgII} that $\displaystyle\frac{1}{4}Z(W)^2 \leq P_e(\pi_{\mathcal{X}},W)\leq (|\mathcal{X}|-1)Z(W)$, where $\pi_{\mathcal{X}}$ is the uniform distribution on $\mathcal{X}$.

An \emph{$(n,M)$-code} $\mathcal{C}$ on the alphabet $\mathcal{X}$  is a subset of $\mathcal{X}^n$ such that $|\mathcal{C}|=M$. The integer $n$ is the \emph{blocklength} of $\mathcal{C}$, and $M$ is the \emph{size} of the code. The \emph{rate} of $\mathcal{C}$ is $\frac{1}{n}\log M$, and it is measured in \emph{nats}. The \emph{error probability of the ML decoder for the code $\mathcal{C}$ when it is used for a channel $W\in\DMC_{\mathcal{X},\mathcal{Y}}$} is given by:
$$P_{e,\mathcal{C}}(W)=1-\frac{1}{|\mathcal{C}|}\sum_{y_1^n\in\mathcal{Y}^n} \max_{x_1^n\in\mathcal{C}}\left\{\prod_{i=1}^n W(y_i|x_i)\right\}.$$

The \emph{optimal error probability of $(n,M)$-codes for a channel $W$} is given by:
$$P_{e,n,M}(W)=\min_{\substack{\mathcal{C}\subset \mathcal{X}^n,\\|\mathcal{C}|=M}}P_{e,\mathcal{C}}(W).$$

The following proposition shows that all the above parameters are continuous:

\begin{myprop}
\label{propContParamDMCXY}
We have:
\begin{itemize}
\item $I:\Delta_{\mathcal{X}}\times\DMC_{\mathcal{X},\mathcal{Y}}\rightarrow \mathbb{R}^+$ is continuous, concave in $p$, and convex in $W$.
\item $C:\DMC_{\mathcal{X},\mathcal{Y}}\rightarrow \mathbb{R}^+$ is continuous and convex.
\item $P_e:\Delta_{\mathcal{X}}\times\DMC_{\mathcal{X},\mathcal{Y}}\rightarrow [0,1]$ is continuous, concave in $p$ and concave in $W$.
\item $Z:\DMC_{\mathcal{X},\mathcal{Y}}\rightarrow [0,1]$ is continuous.
\item For every code $\mathcal{C}$ on $\mathcal{X}$, $P_{e,\mathcal{C}}:\DMC_{\mathcal{X},\mathcal{Y}}\rightarrow [0,1]$ is continuous.
\item For every $n>0$ and every $1\leq M\leq |\mathcal{X}|^n$, the mapping $P_{e,n,M}:\DMC_{\mathcal{X},\mathcal{Y}}\rightarrow [0,1]$ is continuous.
\end{itemize}
\end{myprop}
\begin{proof}
These facts are well known, especially the continuity of $I$, its concavity in $p$, and its convexity in $W$ \cite{Cover}. Since $C$ is the supremum of a family of mappings that are convex in $W$, it is also convex in $W$. For a proof of the continuity of $C$, see Appendix \ref{appContParamDMCXY}. The continuity of $Z$, $P_e$ and $P_{e,\mathcal{C}}$ follows immediately from their definitions. Moreover, since $P_{e,n,M}$ is the minimum of a finite number of continuous mappings, it is continuous. The concavity of $P_e$ in $p$ and in $W$ can also be easily seen from the definition.
\end{proof}

\subsection{Channel operations}

\label{subsecChanOper}

If $W\in\DMC_{\mathcal{X},\mathcal{Y}}$ and $V\in\DMC_{\mathcal{Y},\mathcal{Z}}$, we define the composition $V\circ W\in\DMC_{\mathcal{X},\mathcal{Z}}$ of $W$ and $V$ as follows:
$$(V\circ W)(z|x)=\sum_{y\in\mathcal{Y}}V(z|y)W(y|x).\;\;\forall x\in\mathcal{X},\;\forall z\in\mathcal{Z}.$$

For every function $f:\mathcal{X}\rightarrow \mathcal{Y}$, define the \emph{deterministic channel} $D_f\in\DMC_{\mathcal{X},\mathcal{Y}}$ as follows:
$$D_f(y|x)=\begin{cases}1\quad&\text{if}\;y=f(x),\\0\quad&\text{otherwise}.\end{cases}$$
It is easy to see that if $f:\mathcal{X}\rightarrow \mathcal{Y}$ and $g:\mathcal{Y}\rightarrow \mathcal{Z}$, then $D_g\circ D_f=D_{g\circ f}$.

For every two channels $W_1\in \DMC_{\mathcal{X}_1,\mathcal{Y}_1}$ and $W_2\in \DMC_{\mathcal{X}_2,\mathcal{Y}_2}$, define the \emph{channel sum} $W_1 \oplus W_2\in \DMC_{\mathcal{X}_1\coprod\mathcal{X}_2,\mathcal{Y}_1\coprod\mathcal{Y}_2}$ of $W_1$ and $W_2$ as:
$$(W_1\oplus W_2)(y,i|x,j)=\begin{cases}W_i(y|x)\quad&\text{if}\;i=j,\\
0&\text{otherwise}.\end{cases}$$
$W_1\oplus W_2$ arises when the transmitter has two channels $W_1$ and $W_2$ at his disposal and he can use exactly one of them at each channel use. It is an easy exercise to check that $e^{C(W_1\oplus W_2)}=e^{C(W_1)}+e^{C(W_2)}$ (remember that we compute the mutual information using the natural logarithm).

We define the \emph{channel product} $W_1\otimes W_2\in \DMC_{\mathcal{X}_1\times\mathcal{X}_2,\mathcal{Y}_1\times\mathcal{Y}_2}$ of $W_1$ and $W_2$ as:
$$(W_1\otimes W_2)(y_1,y_2|x_1,x_2)=W_1(y_1|x_1)W_2(y_2|x_2).$$
$W_1\otimes W_2$ arises when the transmitter has two channels $W_1$ and $W_2$ at his disposal and he uses both of them at each channel use. It is an easy exercise to check that $C(W_1\otimes W_2)=C(W_1)+C(W_2)$, or equivalently $e^{C(W_1\otimes W_2)}=e^{C(W_1)}\cdot e^{C(W_2)}$. Channel sums and products were first introduced by Shannon in \cite{ChannelSumProduct}.

For every $W_1\in\DMC_{\mathcal{X},\mathcal{Y}_1}$, $W_2\in\DMC_{\mathcal{X},\mathcal{Y}_2}$ and every $0\leq \alpha\leq 1$, we define the $\alpha$-interpolation $[\alpha W_1,(1-\alpha)W_2]\in\DMC_{\mathcal{X},\mathcal{Y}_1\coprod\mathcal{Y}_2}$ between $W_1$ and $W_2$ as:
$$[\alpha W_1,(1-\alpha)W_2](y,i\big|x)=\begin{cases}\alpha W_1(y|x)\quad&\text{if}\;i=1,\\(1-\alpha) W_2(y|x)\quad&\text{if}\;i=2. \end{cases}$$
Channel interpolation arises when a channel behaves as $W_1$ with probability $\alpha$ and as $W_2$ with probability $1-\alpha$. The transmitter has no control on which behavior the channel chooses, but on the other hand, the receiver knows which one was chosen. Channel interpolations were used in \cite{Interpolation} to construct interpolations between polar codes and Reed-Muller codes.

Now fix a binary operation $\ast$ on $\mathcal{X}$. For every $W\in\DMC_{\mathcal{X},\mathcal{Y}}$, define $W^-\in\DMC_{\mathcal{X},\mathcal{Y}^2}$ and $W^+\in\DMC_{\mathcal{X},\mathcal{Y}^2\times\mathcal{X}}$ as:
$$W^-(y_1,y_2|u_1)=\frac{1}{|\mathcal{X}|}\sum_{u_2\in\mathcal{X}}W(y_1|u_1\ast u_2)W(y_2|u_2),$$
and
$$W^+(y_1,y_2,u_1|u_2)=\frac{1}{|\mathcal{X}|}W(y_1|u_1\ast u_2)W(y_2|u_2).$$
These operations generalize Ar{\i}kan's polarization transformations \cite{Arikan}.

\begin{myprop}
\label{propContOperDMCXY}
We have:
\begin{itemize}
\item The mapping $(W,V)\rightarrow V\circ W$ from $\DMC_{\mathcal{X},\mathcal{Y}}\times \DMC_{\mathcal{Y},\mathcal{Z}}$ to $\DMC_{\mathcal{X},\mathcal{Z}}$ is continuous.
\item The mapping $(W_1,W_2)\rightarrow W_1\oplus W_2$ from $\DMC_{\mathcal{X}_1,\mathcal{Y}_1}\times \DMC_{\mathcal{X}_2,\mathcal{Y}_2}$ to $\DMC_{\mathcal{X}_1\coprod\mathcal{X}_2,\mathcal{Y}_1\coprod\mathcal{Y}_2}$ is continuous.
\item The mapping $(W_1,W_2)\rightarrow W_1\otimes W_2$ from $\DMC_{\mathcal{X}_1,\mathcal{Y}_1}\times \DMC_{\mathcal{X}_2,\mathcal{Y}_2}$ to $\DMC_{\mathcal{X}_1\times\mathcal{X}_2,\mathcal{Y}_1\times\mathcal{Y}_2}$ is continuous.
\item The mapping $(W_1,W_2,\alpha)\rightarrow [\alpha W_1,(1-\alpha)W_2]$ from $\DMC_{\mathcal{X},\mathcal{Y}_1}\times \DMC_{\mathcal{X},\mathcal{Y}_2}\times[0,1]$ to $\DMC_{\mathcal{X},\mathcal{Y}_1\coprod\mathcal{Y}_2}$ is continuous.
\item For any binary operation $\ast$ on $\mathcal{X}$, the mapping $W\rightarrow W^-$ from $\DMC_{\mathcal{X},\mathcal{Y}}$ to $\DMC_{\mathcal{X},\mathcal{Y}^2}$ is continuous.
\item For any binary operation $\ast$ on $\mathcal{X}$, the mapping $W\rightarrow W^+$ from $\DMC_{\mathcal{X},\mathcal{Y}}$ to $\DMC_{\mathcal{X},\mathcal{Y}^2\times\mathcal{X}}$ is continuous.
\end{itemize}
\end{myprop}
\begin{proof}
The continuity immediately follows from the definitions.
\end{proof}

\section{Continuity on $\DMC_{\mathcal{X},\mathcal{Y}}^{(o)}$}

\label{secContDMCXYo}

It is well known that the parameters defined in section \ref{subsecChanParam} depend only on the $R_{\mathcal{X},\mathcal{Y}}^{(o)}$-equivalence class of $W$. Therefore, we can define those parameters for any $\hat{W}\in\DMC_{\mathcal{X},\mathcal{Y}}^{(o)}$ through the transcendent mapping (defined in Lemma \ref{lemQuotientFunction}). The following proposition shows that those parameters are continuous on $\DMC_{\mathcal{X},\mathcal{Y}}^{(o)}$:

\begin{myprop}
\label{propContParamDMCXYo}
We have:
\begin{itemize}
\item $I:\Delta_{\mathcal{X}}\times\DMC_{\mathcal{X},\mathcal{Y}}^{(o)}\rightarrow \mathbb{R}^+$ is continuous and concave in $p$.
\item $C:\DMC_{\mathcal{X},\mathcal{Y}}^{(o)}\rightarrow \mathbb{R}^+$ is continuous.
\item $P_e:\Delta_{\mathcal{X}}\times\DMC_{\mathcal{X},\mathcal{Y}}^{(o)}\rightarrow [0,1]$ is continuous and concave in $p$.
\item $Z:\DMC_{\mathcal{X},\mathcal{Y}}^{(o)}\rightarrow [0,1]$ is continuous.
\item For every code $\mathcal{C}$ on $\mathcal{X}$, $P_{e,\mathcal{C}}:\DMC_{\mathcal{X},\mathcal{Y}}^{(o)}\rightarrow [0,1]$ is continuous.
\item For every $n>0$ and every $1\leq M\leq |\mathcal{X}|^n$, the mapping $P_{e,n,M}:\DMC_{\mathcal{X},\mathcal{Y}}^{(o)}\rightarrow [0,1]$ is continuous.
\end{itemize}
\end{myprop}
\begin{proof}
Since the corresponding parameters are continuous on $\DMC_{\mathcal{X},\mathcal{Y}}$ (Proposition \ref{propContParamDMCXY}), Lemma \ref{lemQuotientFunction} implies that they are continuous on $\DMC_{\mathcal{X},\mathcal{Y}}^{(o)}$. The only cases that need a special treatment are those of $I$ and $Z$. We will only prove the continuity of $I$ since the proof of continuity of $Z$ is similar.

Define the relation $R$ on $\Delta_{\mathcal{X}}\times \DMC_{\mathcal{X},\mathcal{Y}}$ as
$$(p_1,W_1)R(p_2,W_2)\;\;\Leftrightarrow\;\; p_1=p_2\;\text{and}\;W_1 R_{\mathcal{X},\mathcal{Y}}^{(o)} W_2.$$
It is easy to see that $I(p,W)$ depends only on the $R$-equivalence class of $(p,W)$. Since $I$ is continuous on $\Delta_{\mathcal{X}}\times \DMC_{\mathcal{X},\mathcal{Y}}$, Lemma \ref{lemQuotientFunction} implies that the transcendent mapping of $I$ is continuous on $(\Delta_{\mathcal{X}}\times \DMC_{\mathcal{X},\mathcal{Y}})/R$. On the other hand, since $\Delta_{\mathcal{X}}$ is locally compact, Theorem \ref{theQuotientProd} implies that $(\Delta_{\mathcal{X}}\times \DMC_{\mathcal{X},\mathcal{Y}})/R$ can be identified with $\Delta_{\mathcal{X}}\times(\DMC_{\mathcal{X},\mathcal{Y}}/R_{\mathcal{X},\mathcal{Y}}^{(o)})=\Delta_{\mathcal{X}}\times\DMC_{\mathcal{X},\mathcal{Y}}^{(o)}$ and the two spaces have the same topology. Therefore, $I$ is continuous on $\Delta_{\mathcal{X}}\times\DMC_{\mathcal{X},\mathcal{Y}}^{(o)}$.
\end{proof}

\vspace*{3mm}

With the exception of channel composition, all the channel operations that were defined in Section \ref{subsecChanOper} can also be ``quotiented". We just need to realize that the equivalence class of the resulting channel depends only on the equivalence classes of the channels that were used in the operation. Let us illustrate this in the case of channel sums:

Let $W_1,W_1'\in \DMC_{\mathcal{X}_1,\mathcal{Y}_1}$ and $W_2,W_2'\in \DMC_{\mathcal{X}_2,\mathcal{Y}_2}$ and assume that $W_1$ is degraded from $W_1'$ and $W_2$ is degraded from $W_2'$. There exists $V_1\in\DMC_{\mathcal{Y}_1,\mathcal{Y}_1}$ and $V_2\in\DMC_{\mathcal{Y}_2,\mathcal{Y}_2}$ such that $W_1=V_1\circ W_1'$ and $W_2=V_2\circ W_2'$. It is easy to see that $W_1\oplus W_2=(V_1\oplus V_2)\circ(W_1'\oplus W_2')$, which shows that $W_1\oplus W_2$ is degraded from $W_1'\oplus W_2'$. This was proved by Shannon in \cite{ShannonDegrad}.

Therefore, if $W_1$ is equivalent to $W_1'$ and $W_2$ is equivalent to $W_2'$, then $W_1\oplus W_2$ is equivalent to $W_1'\oplus W_2'$. This allows us to define the channel sum for every $\hat{W}_1\in\DMC_{\mathcal{X}_1,\mathcal{Y}_1}^{(o)}$ and every $\overline{W}_2\in\DMC_{\mathcal{X}_2,\mathcal{Y}_2}^{(o)}$ as $\hat{W}_1\oplus \overline{W}_2 = \widetilde{W_1'\oplus W_2'}\in \DMC_{\mathcal{X}_1\coprod\mathcal{X}_2,\mathcal{Y}_1\coprod\mathcal{Y}_2}^{(o)}$ for any $W_1'\in\hat{W}_1$ and any $W_2'\in\overline{W}_2$, where $\widetilde{W_1'\oplus W_2'}$ is the $R_{\mathcal{X}_1\coprod\mathcal{X}_2,\mathcal{Y}_1\coprod\mathcal{Y}_2}^{(o)}$-equivalence class of $W_1'\oplus W_2'$.

With the exception of channel composition, we can ``quotient" all the channel operations of Section \ref{subsecChanOper} in a similar fashion. Moreover, we can show that they are continuous:

\begin{myprop}
\label{propContOperDMCXYo}
We have:
\begin{itemize}
\item The mapping $(\hat{W}_1,\overline{W}_2)\rightarrow \hat{W}_1\oplus \overline{W}_2$ from $\DMC_{\mathcal{X}_1,\mathcal{Y}_1}^{(o)}\times \DMC_{\mathcal{X}_2,\mathcal{Y}_2}^{(o)}$ to $\DMC_{\mathcal{X}_1\coprod\mathcal{X}_2,\mathcal{Y}_1\coprod\mathcal{Y}_2}^{(o)}$ is continuous.
\item The mapping $(\hat{W}_1,\overline{W}_2)\rightarrow \hat{W}_1\otimes \overline{W}_2$ from $\DMC_{\mathcal{X}_1,\mathcal{Y}_1}^{(o)}\times \DMC_{\mathcal{X}_2,\mathcal{Y}_2}^{(o)}$ to $\DMC_{\mathcal{X}_1\times\mathcal{X}_2,\mathcal{Y}_1\times\mathcal{Y}_2}^{(o)}$ is continuous.
\item The mapping $(\hat{W}_1,\overline{W}_2,\alpha)\rightarrow [\alpha \hat{W}_1,(1-\alpha)\overline{W}_2]$ from $\DMC_{\mathcal{X},\mathcal{Y}_1}^{(o)}\times \DMC_{\mathcal{X},\mathcal{Y}_2}^{(o)}\times[0,1]$ to $\DMC_{\mathcal{X},\mathcal{Y}_1\coprod\mathcal{Y}_2}^{(o)}$ is continuous.
\item For any binary operation $\ast$ on $\mathcal{X}$, the mapping $\hat{W}\rightarrow \hat{W}^-$ from $\DMC_{\mathcal{X},\mathcal{Y}}^{(o)}$ to $\DMC_{\mathcal{X},\mathcal{Y}^2}^{(o)}$ is continuous.
\item For any binary operation $\ast$ on $\mathcal{X}$, the mapping $\hat{W}\rightarrow \hat{W}^+$ from $\DMC_{\mathcal{X},\mathcal{Y}}^{(o)}$ to $\DMC_{\mathcal{X},\mathcal{Y}^2\times\mathcal{X}}^{(o)}$ is continuous.
\end{itemize}
\end{myprop}
\begin{proof}
We only prove the continuity of the channel sum because the proof of continuity of the other operations is similar.

Let $\Proj:\DMC_{\mathcal{X}_1\coprod\mathcal{X}_2,\mathcal{Y}_1\coprod\mathcal{Y}_2}\rightarrow \DMC_{\mathcal{X}_1\coprod\mathcal{X}_2,\mathcal{Y}_1\coprod\mathcal{Y}_2}^{(o)}$ be the projection onto the $R_{\mathcal{X}_1\coprod\mathcal{X}_2,\mathcal{Y}_1\coprod\mathcal{Y}_2}^{(o)}$-equivalence classes. Define the mapping $f:\DMC_{\mathcal{X}_1,\mathcal{Y}_1}\times \DMC_{\mathcal{X}_2,\mathcal{Y}_2}\rightarrow \DMC_{\mathcal{X}_1\coprod\mathcal{X}_2,\mathcal{Y}_1\coprod\mathcal{Y}_2}^{(o)}$ as $f(W_1,W_2)=\Proj(W_1\oplus W_2)$. Clearly, $f$ is continuous.

Now define the equivalence relation $R$ on $\DMC_{\mathcal{X}_1,\mathcal{Y}_1}\times \DMC_{\mathcal{X}_2,\mathcal{Y}_2}$ as:
$$(W_1,W_2)R(W_1',W_2')\;\;\Leftrightarrow\;\; W_1 R_{\mathcal{X}_1,\mathcal{Y}_1}^{(o)}W_1'\;\text{and}\;W_2 R_{\mathcal{X}_2,\mathcal{Y}_2}^{(o)}W_2'.$$
The discussion before the proposition shows that $f(W_1,W_2)=\Proj(W_1\oplus W_2)$ depends only on the $R$-equivalence class of $(W_1,W_2)$. Lemma \ref{lemQuotientFunction} now shows that the transcendent map of $f$ defined on $(\DMC_{\mathcal{X}_1,\mathcal{Y}_1}\times \DMC_{\mathcal{X}_2,\mathcal{Y}_2})/R$ is continuous.

Notice that $(\DMC_{\mathcal{X}_1,\mathcal{Y}_1}\times \DMC_{\mathcal{X}_2,\mathcal{Y}_2})/R$ can be identified with $\DMC_{\mathcal{X}_1,\mathcal{Y}_1}^{(o)}\times \DMC_{\mathcal{X}_2,\mathcal{Y}_2}^{(o)}$. Therefore, we can define $f$ on $\DMC_{\mathcal{X}_1,\mathcal{Y}_1}^{(o)}\times \DMC_{\mathcal{X}_2,\mathcal{Y}_2}^{(o)}$ through this identification. Moreover, since $\DMC_{\mathcal{X}_1,\mathcal{Y}_1}$ and $\DMC_{\mathcal{X}_2,\mathcal{Y}_2}^{(o)}$ are locally compact and Hausdorff, Corollary \ref{corQuotientProd} implies that the canonical bijection between $(\DMC_{\mathcal{X}_1,\mathcal{Y}_1}\times \DMC_{\mathcal{X}_2,\mathcal{Y}_2})/R$ and $\DMC_{\mathcal{X}_1,\mathcal{Y}_1}^{(o)}\times \DMC_{\mathcal{X}_2,\mathcal{Y}_2}^{(o)}$ is a homeomorphism. 

Now since the mapping $f$ on $\DMC_{\mathcal{X}_1,\mathcal{Y}_1}^{(o)}\times \DMC_{\mathcal{X}_2,\mathcal{Y}_2}^{(o)}$ is just the channel sum, we conclude that the mapping $(\hat{W}_1,\overline{W}_2)\rightarrow \hat{W}_1\oplus \overline{W}_2$ from $\DMC_{\mathcal{X}_1,\mathcal{Y}_1}^{(o)}\times \DMC_{\mathcal{X}_2,\mathcal{Y}_2}^{(o)}$ to $\DMC_{\mathcal{X}_1\coprod\mathcal{X}_2,\mathcal{Y}_1\coprod\mathcal{Y}_2}^{(o)}$ is continuous.
\end{proof}

\section{Continuity in the strong topology}

\label{secContMappingsStrong}

The following lemma provides a way to check whether a mapping defined on $(\DMC_{\mathcal{X},\ast}^{(o)},\mathcal{T}_{s,\mathcal{X},\ast}^{(o)})$ is continuous:

\begin{mylem}
\label{lemContinuityForStrongTop}
Let $(S,\mathcal{V})$ be an arbitrary topological space. A mapping $f:\DMC_{\mathcal{X},\ast}^{(o)}\rightarrow S$ is continuous on $(\DMC_{\mathcal{X},\ast}^{(o)},\mathcal{T}_{s,\mathcal{X},\ast}^{(o)})$ if and only if it is continuous on $(\DMC_{\mathcal{X},[n]}^{(o)},\mathcal{T}_{\mathcal{X},[n]}^{(o)})$ for every $n\geq 1$.
\end{mylem}
\begin{proof}
\begin{align*}
\textstyle f\;\text{is continuous on}\;(\DMC_{\mathcal{X},\ast}^{(o)},\mathcal{T}_{s,\mathcal{X},\ast}^{(o)})\;\;&\textstyle\Leftrightarrow\;\; f^{-1}(V)\in \mathcal{T}_{s,\mathcal{X},\ast}^{(o)}\;\;\forall V\in\mathcal{V}\\
&\textstyle\Leftrightarrow\;\; f^{-1}(V)\cap \DMC_{\mathcal{X},[n]}^{(o)} \in \mathcal{T}_{\mathcal{X},[n]}^{(o)}\;\;\forall n\geq 1,\; \forall V\in\mathcal{V}\\
&\textstyle\Leftrightarrow\;\; f\;\text{is continuous on}\; (\DMC_{\mathcal{X},[n]}^{(o)},\mathcal{T}_{\mathcal{X},[n]}^{(o)})\;\;\forall n\geq 1.
\end{align*}
\end{proof}

\vspace*{3mm}

Since the channel parameters $I$, $C$, $P_e$, $Z$, $P_{e,\mathcal{C}}$ and $P_{e,n,M}$ are defined on $\DMC_{\mathcal{X},[n]}^{(o)}$ for every $n\geq 1$ (see Section \ref{secContDMCXYo}), they are also defined on $\DMC_{\mathcal{X},\ast}^{(o)}={\displaystyle\bigcup_{n\geq 1}}\DMC_{\mathcal{X},[n]}^{(o)}$. The following proposition shows that those parameters are continuous in the strong topology:

\begin{myprop}
\label{propContParamDMCXoStr}
Let $\mathcal{U}_{\mathcal{X}}$ be the standard topology on $\Delta_{\mathcal{X}}$. We have:
\begin{itemize}
\item $I:\Delta_{\mathcal{X}}\times\DMC_{\mathcal{X},\ast}^{(o)}\rightarrow \mathbb{R}^+$ is continuous on $(\Delta_{\mathcal{X}}\times\DMC_{\mathcal{X},\ast}^{(o)}, \mathcal{U}_{\mathcal{X}}\otimes\mathcal{T}_{s,\mathcal{X},\ast}^{(o)})$ and concave in $p$.
\item $C:\DMC_{\mathcal{X},\ast}^{(o)}\rightarrow \mathbb{R}^+$ is continuous on $(\DMC_{\mathcal{X},\ast}^{(o)},\mathcal{T}_{s,\mathcal{X},\ast}^{(o)})$.
\item $P_e:\Delta_{\mathcal{X}}\times\DMC_{\mathcal{X},\ast}^{(o)}\rightarrow [0,1]$ is continuous on $(\Delta_{\mathcal{X}}\times\DMC_{\mathcal{X},\ast}^{(o)}, \mathcal{U}_{\mathcal{X}}\otimes\mathcal{T}_{s,\mathcal{X},\ast}^{(o)})$ and concave in $p$.
\item $Z:\DMC_{\mathcal{X},\ast}^{(o)}\rightarrow [0,1]$ is continuous on $(\DMC_{\mathcal{X},\ast}^{(o)},\mathcal{T}_{s,\mathcal{X},\ast}^{(o)})$.
\item For every code $\mathcal{C}$ on $\mathcal{X}$, $P_{e,\mathcal{C}}:\DMC_{\mathcal{X},\ast}^{(o)}\rightarrow [0,1]$ is continuous on $(\DMC_{\mathcal{X},\ast}^{(o)},\mathcal{T}_{s,\mathcal{X},\ast}^{(o)})$.
\item For every $n>0$ and every $1\leq M\leq |\mathcal{X}|^n$, the mapping $P_{e,n,M}:\DMC_{\mathcal{X},\ast}^{(o)}\rightarrow [0,1]$ is continuous on $(\DMC_{\mathcal{X},\ast}^{(o)},\mathcal{T}_{s,\mathcal{X},\ast}^{(o)})$.
\end{itemize}
\end{myprop}
\begin{proof}
The continuity of $C, Z, P_{e,C}$ and $P_{e,n,M}$ immediately follows from Proposition \ref{propContParamDMCXYo} and Lemma \ref{lemContinuityForStrongTop}. Since the proofs of continuity of $I$ and $Z$ are similar, we only prove the continuity for $I$.

Due to the distributivity of the product with respect to disjoint unions, we have  
\begin{align*}
\textstyle\Delta_{\mathcal{X}}\times\DMC_{\mathcal{X},\ast}={\displaystyle\coprod_{n\geq 1}}\left(\Delta_{\mathcal{X}}\times \DMC_{\mathcal{X},[n]}\right),
\end{align*}
and
\begin{align*}
\mathcal{U}_{\mathcal{X}}\otimes\mathcal{T}_{s,\mathcal{X},\ast}={\displaystyle\bigoplus_{n\geq 1}}\left(\mathcal{U}_{\mathcal{X}}\otimes \mathcal{T}_{\mathcal{X},[n]}\right).
\end{align*}
Therefore, $(\Delta_{\mathcal{X}}\times\DMC_{\mathcal{X},\ast},\mathcal{U}_{\mathcal{X}}\otimes\mathcal{T}_{s,\mathcal{X},\ast})$ is the disjoint union of the spaces $(\Delta_{\mathcal{X}}\times \DMC_{\mathcal{X},[n]})_{n\geq 1}$. Moreover, $I$ is continuous on $\Delta_{\mathcal{X}}\times \DMC_{\mathcal{X},[n]}$ for every $n\geq 1$. We conclude that $I$ is continuous on $(\Delta_{\mathcal{X}}\times\DMC_{\mathcal{X},\ast},\mathcal{U}_{\mathcal{X}}\otimes\mathcal{T}_{s,\mathcal{X},\ast})$.

Define the relation $R$ on $\Delta_{\mathcal{X}}\times\DMC_{\mathcal{X},\ast}$ as follows: $(p_1,W_1)R(p_2,W_2)$ if and only if $p_1=p_2$ and $W_1 R_{\mathcal{X},\ast}^{(o)} W_2$. Since $I(p,W)$ depends only on the $R$-equivalence class of $(p,W)$, Lemma \ref{lemQuotientFunction} shows that the transcendent map of $I$ is a continuous mapping from $\big((\Delta_{\mathcal{X}}\times\DMC_{\mathcal{X},\ast})/R,(\mathcal{U}_{\mathcal{X}}\otimes\mathcal{T}_{s,\mathcal{X},\ast})/R\big)$ to $\mathbb{R}^+$. On the other hand, since $\Delta_{\mathcal{X}}$ is locally compact and Hausdorff, Theorem \ref{theQuotientProd} implies that $\big((\Delta_{\mathcal{X}}\times\DMC_{\mathcal{X},\ast})/R,(\mathcal{U}_{\mathcal{X}}\otimes\mathcal{T}_{s,\mathcal{X},\ast})/R\big)$ can be identified with $\big(\Delta_{\mathcal{X}}\times(\DMC_{\mathcal{X},\ast}/R_{\mathcal{X},\ast}^{(o)}), \mathcal{U}_{\mathcal{X}}\otimes(\mathcal{T}_{s,\mathcal{X},\ast}/R_{\mathcal{X},\ast}^{(o)})\big)=(\Delta_{\mathcal{X}}\times\DMC_{\mathcal{X},\ast}^{(o)}, \mathcal{U}_{\mathcal{X}}\otimes\mathcal{T}_{s,\mathcal{X},\ast}^{(o)})$. Therefore, $I$ is continuous on $(\Delta_{\mathcal{X}}\times\DMC_{\mathcal{X},\ast}^{(o)}, \mathcal{U}_{\mathcal{X}}\otimes\mathcal{T}_{s,\mathcal{X},\ast}^{(o)})$.
\end{proof}

\vspace*{3mm}

It is also possible to extend the definition of all the channel operations that were defined in section \ref{secContDMCXYo} to $\DMC_{\mathcal{X},\ast}^{(o)}$. Moreover, it is possible to show that many channel operations are continuous in the strong topology:

\begin{myprop}
\label{propContOperDMCXoStr}
Assume that all equivalent channel spaces are endowed with the strong topology. We have:
\begin{itemize}
\item The mapping $(\hat{W}_1,\overline{W}_2)\rightarrow \hat{W}_1\oplus \overline{W}_2$ from $\DMC_{\mathcal{X}_1,\ast}^{(o)}\times \DMC_{\mathcal{X}_2,\mathcal{Y}_2}^{(o)}$ to $\DMC_{\mathcal{X}_1\coprod\mathcal{X}_2,\ast}^{(o)}$ is continuous.
\item The mapping $(\hat{W}_1,\overline{W}_2)\rightarrow \hat{W}_1\otimes \overline{W}_2$ from $\DMC_{\mathcal{X}_1,\ast}^{(o)}\times \DMC_{\mathcal{X}_2,\mathcal{Y}_2}^{(o)}$ to $\DMC_{\mathcal{X}_1\times\mathcal{X}_2,\ast}^{(o)}$ is continuous.
\item The mapping $(\hat{W}_1,\overline{W}_2,\alpha)\rightarrow [\alpha \hat{W}_1,(1-\alpha)\overline{W}_2]$ from $\DMC_{\mathcal{X},\ast}\times \DMC_{\mathcal{X},\mathcal{Y}_2}^{(o)}\times[0,1]$ to $\DMC_{\mathcal{X},\ast}^{(o)}$ is continuous.
\item For any binary operation $\ast$ on $\mathcal{X}$, the mapping $\hat{W}\rightarrow \hat{W}^-$ from $\DMC_{\mathcal{X},\ast}^{(o)}$ to $\DMC_{\mathcal{X},\ast}^{(o)}$ is continuous.
\item For any binary operation $\ast$ on $\mathcal{X}$, the mapping $\hat{W}\rightarrow \hat{W}^+$ from $\DMC_{\mathcal{X},\ast}^{(o)}$ to $\DMC_{\mathcal{X},\ast}^{(o)}$ is continuous.
\end{itemize}
\end{myprop}
\begin{proof}
We only prove the continuity of the channel interpolation because the proof of the continuity of other operations is similar.

Let $\mathcal{U}$ be the standard topology on $[0,1]$. Due to the distributivity of the product with respect to disjoint unions, we have:
$$\textstyle\DMC_{\mathcal{X},\ast}\times\DMC_{\mathcal{X},\mathcal{Y}_2}\times [0,1]={\displaystyle\coprod_{n\geq1}}(\DMC_{\mathcal{X},[n]}\times\DMC_{\mathcal{X},\mathcal{Y}_2}\times [0,1]),$$
and
$$\textstyle\mathcal{T}_{s,\mathcal{X},\ast}\otimes\mathcal{T}_{\mathcal{X},\mathcal{Y}_2}\otimes \mathcal{U}={\displaystyle\bigoplus_{n\geq1}}\left(\mathcal{T}_{\mathcal{X},[n]}\otimes\mathcal{T}_{\mathcal{X},\mathcal{Y}_2}\otimes \mathcal{U}\right).$$

Therefore, the space $\DMC_{\mathcal{X},\ast}\times\DMC_{\mathcal{X},\mathcal{Y}_2}\times [0,1]$ is the topological disjoint union of the spaces $(\DMC_{\mathcal{X},[n]}\times\DMC_{\mathcal{X},\mathcal{Y}_2}\times [0,1])_{n\geq 1}$.

For every $n\geq 1$, let $\Proj_n$ be the projection onto the $R_{\mathcal{X},[n]\coprod\mathcal{Y}_2}^{(o)}$-equivalence classes and let $i_n$ be the canonical injection from $\DMC_{\mathcal{X},[n]\coprod\mathcal{Y}_2}^{(o)}$ to $\DMC_{\mathcal{X},\ast}^{(o)}$. 

Define the mapping $f: \DMC_{\mathcal{X},\ast}\times\DMC_{\mathcal{X},\mathcal{Y}_2}\times [0,1]\rightarrow \DMC_{\mathcal{X},\ast}^{(o)}$ as $$\textstyle f(W_1,W_2,\alpha)=i_n(\Proj_n([\alpha W_1,(1-\alpha)W_2]))=[\alpha\hat{W}_1,(1-\alpha)\overline{W}_2],$$
where $n$ is the unique integer satisfying $W_1\in \DMC_{\mathcal{X},[n]}$. $\hat{W}_1$ and $\overline{W}_2$ are the $R_{\mathcal{X},[n]}^{(o)}$ and $R_{\mathcal{X},\mathcal{Y}_2}^{(o)}$-equivalence classes of $W_1$ and $W_2$ respectively.

Due to Proposition \ref{propContOperDMCXY} and due to the continuity of $\Proj_n$ and $i_n$, the mapping $f$ is continuous on $\DMC_{\mathcal{X},[n]}\times\DMC_{\mathcal{X},\mathcal{Y}_2}\times [0,1]$ for every $n\geq 1$. Therefore, $f$ is continuous on $(\DMC_{\mathcal{X},\ast}\times\DMC_{\mathcal{X},\mathcal{Y}_2}\times [0,1],\mathcal{T}_{s,\mathcal{X},\ast}\otimes\mathcal{T}_{\mathcal{X},\mathcal{Y}_2}\otimes \mathcal{U})$.

Let $R'$ be the equivalence relation defined on $\DMC_{\mathcal{X},\ast}\times\DMC_{\mathcal{X},\mathcal{Y}_2}$ as follows: $(W_1,W_2)R'(W_1',W_2')$ if and only if $W_1 R_{\mathcal{X},\ast}^{(o)} W_1'$ and $W_2 R_{\mathcal{X},\mathcal{Y}_2}^{(o)} W_2'$. Also, define the equivalence relation $R$ on $\DMC_{\mathcal{X},\ast}\times\DMC_{\mathcal{X},\mathcal{Y}_2}\times [0,1]$ as follows: $(W_1,W_2,\alpha)R(W_1',W_2',\alpha')$ if and only if $(W_1,W_2)R'(W_1',W_2')$ and $\alpha=\alpha'$.

Since $f(W_1,W_2,\alpha)$ depends only on the $R$-equivalence class of $(W_1,W_2,\alpha)$, Lemma \ref{lemQuotientFunction} implies that the transcendent mapping of $f$ is continuous on $(\DMC_{\mathcal{X},\ast}\times\DMC_{\mathcal{X},\mathcal{Y}_2}\times [0,1])/R$.

Since $[0,1]$ is Hausdorff and locally compact, Theorem \ref{theQuotientProd} implies that the canonical bijection from $(\DMC_{\mathcal{X},\ast}\times\DMC_{\mathcal{X},\mathcal{Y}_2}\times [0,1])/R$ to $\big((\DMC_{\mathcal{X},\ast}\times\DMC_{\mathcal{X},\mathcal{Y}_2})/R'\big)\times [0,1])$ is a homeomorphism. On the other hand, since $(\DMC_{\mathcal{X},\ast},\mathcal{T}_{s,\mathcal{X},\ast})$ and $\DMC_{\mathcal{X},\mathcal{Y}_2}^{(o)}=\DMC_{\mathcal{X},\mathcal{Y}_2}/R_{\mathcal{X},\mathcal{Y}_2}^{(o)}$ are Hausdorff and locally compact, Corollary \ref{corQuotientProd} implies that the canonical bijection from $\DMC_{\mathcal{X},\ast}^{(o)}\times \DMC_{\mathcal{X},\mathcal{Y}_2}^{(o)}$ to $(\DMC_{\mathcal{X},\ast}\times \DMC_{\mathcal{X},\mathcal{Y}_2})/R'$ is a homeomorphism. We conclude that the channel interpolation is continuous on $(\DMC_{\mathcal{X},\ast}^{(o)}\times \DMC_{\mathcal{X},\mathcal{Y}_2}^{(o)}\times[0,1],\mathcal{T}_{s,\mathcal{X},\ast}^{(o)}\otimes\mathcal{T}_{\mathcal{X},\mathcal{Y}}^{(o)}\otimes\mathcal{U})$.
\end{proof}

\begin{mycor}
\label{corStrongContracStrong}
$(\DMC_{\mathcal{X},\ast}^{(o)},\mathcal{T}_{s,\mathcal{X},\ast}^{(o)})$ is strongly contractible to every point in $\DMC_{\mathcal{X},\ast}^{(o)}$.
\end{mycor}
\begin{proof}
Fix $\hat{W}_0\in \DMC_{\mathcal{X},\ast}^{(o)}$. Define the mapping $H:\DMC_{\mathcal{X},\ast}^{(o)}\times[0,1]\rightarrow \DMC_{\mathcal{X},\ast}^{(o)}$ as $H(\hat{W},\alpha)=[\alpha\hat{W}_0,(1-\alpha)\hat{W}]$. $H$ is continuous by Proposition \ref{propContOperDMCXoStr}. We also have $H(\hat{W},0)=\hat{W}$ and $H(\hat{W},1)=\hat{W}_0$ for every $\hat{W}\in \DMC_{\mathcal{X},\ast}^{(o)}$. Moreover, $H(\hat{W}_0,\alpha)=\hat{W}_0$ for every $0\leq \alpha\leq 1$. Therefore, $(\DMC_{\mathcal{X},\ast}^{(o)},\mathcal{T}_{s,\mathcal{X},\ast}^{(o)})$ is strongly contractible to every point in $\DMC_{\mathcal{X},\ast}^{(o)}$.
\end{proof}

\vspace*{3mm}

The reader might be wondering why channel operations such as the channel sum were not shown to be continuous on the whole space $\DMC_{\mathcal{X}_1,\ast}^{(o)}\times \DMC_{\mathcal{X}_2,\ast}^{(o)}$ instead of the smaller space $\DMC_{\mathcal{X}_1,\ast}^{(o)}\times \DMC_{\mathcal{X}_2,\mathcal{Y}_2}^{(o)}$. The reason is because we cannot apply Corollary \ref{corQuotientProd} to $\DMC_{\mathcal{X}_1,\ast}\times \DMC_{\mathcal{X}_2,\ast}$ and $\DMC_{\mathcal{X}_1,\ast}^{(o)}\times \DMC_{\mathcal{X}_2,\ast}^{(o)}$ since neither $\DMC_{\mathcal{X}_1,\ast}^{(o)}$ nor $\DMC_{\mathcal{X}_2,\ast}^{(o)}$ is locally compact (under the strong topology).

One potential method to show the continuity of the channel sum on $(\DMC_{\mathcal{X}_1,\ast}^{(o)}\times\DMC_{\mathcal{X}_2,\ast}^{(o)},\mathcal{T}_{s,\mathcal{X}_1,\ast}^{(o)}\otimes \mathcal{T}_{s,\mathcal{X}_2,\ast}^{(o)})$ is as follows: let $R$ be the equivalence relation on $\DMC_{\mathcal{X}_1,\ast}\times\DMC_{\mathcal{X}_2,\ast}$ defined as $(W_1,W_2)R(W_1',W_2')$ if and only if $W_1 R_{\mathcal{X}_1,\ast}^{(o)}W_1'$ and $W_2 R_{\mathcal{X}_2,\ast}^{(o)}W_2'$. We can identify $(\DMC_{\mathcal{X}_1,\ast}\times\DMC_{\mathcal{X}_2,\ast})/R$ with $\DMC_{\mathcal{X}_1,\ast}^{(o)}\times\DMC_{\mathcal{X}_2,\ast}^{(o)}$ through the canonical bijection. Using Lemma \ref{lemQuotientFunction}, it is easy to see that the mapping $(\hat{W}_1,\overline{W}_2)\rightarrow \hat{W}_1\oplus\overline{W}_2$ is continuous from $\big(\DMC_{\mathcal{X}_1,\ast}^{(o)}\times\DMC_{\mathcal{X}_2,\ast}^{(o)}, (\mathcal{T}_{s,\mathcal{X}_1,\ast}\otimes \mathcal{T}_{s,\mathcal{X}_2,\ast})/R\big)$ to $(\DMC_{\mathcal{X}_1\coprod\mathcal{X}_2,\ast}^{(o)},\mathcal{T}_{s,\mathcal{X}_1\coprod\mathcal{X}_2,\ast}^{(o)})$.

It was shown in \cite{CompactlyGenerated} that the topology $(\mathcal{T}_{s,\mathcal{X}_1,\ast}\otimes \mathcal{T}_{s,\mathcal{X}_2,\ast})/R$ is homeomorphic to $\kappa(\mathcal{T}_{s,\mathcal{X}_1,\ast}^{(o)}\otimes \mathcal{T}_{s,\mathcal{X}_2,\ast}^{(o)})$ through the canonical bijection, where $\kappa(\mathcal{T}_{s,\mathcal{X}_1,\ast}^{(o)}\otimes \mathcal{T}_{s,\mathcal{X}_2,\ast}^{(o)})$ is the coarsest topology that is both compactly generated and finer than $\mathcal{T}_{s,\mathcal{X}_1,\ast}^{(o)}\otimes \mathcal{T}_{s,\mathcal{X}_2,\ast}^{(o)}$. Therefore, the mapping $(\hat{W}_1,\overline{W}_2)\rightarrow \hat{W}_1\oplus\overline{W}_2$ is continuous on $\big(\DMC_{\mathcal{X}_1,\ast}^{(o)}\times\DMC_{\mathcal{X}_2,\ast}^{(o)}, \kappa(\mathcal{T}_{s,\mathcal{X}_1,\ast}^{(o)}\otimes \mathcal{T}_{s,\mathcal{X}_2,\ast}^{(o)})\big)$. This means that if $\mathcal{T}_{s,\mathcal{X}_1,\ast}^{(o)}\otimes \mathcal{T}_{s,\mathcal{X}_2,\ast}^{(o)}$ is compactly generated, we will have $\mathcal{T}_{s,\mathcal{X}_1,\ast}^{(o)}\otimes \mathcal{T}_{s,\mathcal{X}_2,\ast}^{(o)}=\kappa(\mathcal{T}_{s,\mathcal{X}_1,\ast}^{(o)}\otimes \mathcal{T}_{s,\mathcal{X}_2,\ast}^{(o)})$ and so the channel sum will be continuous on $(\DMC_{\mathcal{X}_1,\ast}^{(o)}\times\DMC_{\mathcal{X}_2,\ast}^{(o)}, \mathcal{T}_{s,\mathcal{X}_1,\ast}^{(o)}\otimes \mathcal{T}_{s,\mathcal{X}_2,\ast}^{(o)})$. Note that although $\mathcal{T}_{s,\mathcal{X}_1,\ast}^{(o)}$ and $\mathcal{T}_{s,\mathcal{X}_2,\ast}^{(o)}$ are compactly generated, their product $\mathcal{T}_{s,\mathcal{X}_1,\ast}^{(o)}\otimes \mathcal{T}_{s,\mathcal{X}_2,\ast}^{(o)}$ might not be compactly generated.

\section{Continuity in the noisiness/weak-$\ast$ and the total variation topologies}

We need to express the channel parameters and operations in terms of the Blackwell measures.

\subsection{Channel parameters}

The following proposition shows that many channel parameters can be expressed as an integral of a continuous function with respect to the Blackwell measure:

\begin{myprop}
\label{propFormulasParamMetaProb}
For every $\hat{W}\in \DMC_{\mathcal{X},\ast}^{(o)}$, we have:
$$\forall p\in\Delta_{\mathcal{X}},\; I(p,\hat{W})=H(p)-|\mathcal{X}|\cdot\int_{\Delta_{\mathcal{X}}}\left(\sum_{x\in\mathcal{X}} p(x)p'(x)\log\frac{p(x)p'(x)}{\displaystyle\sum_{x'}p(x')p'(x')} \right)\cdot d{\MP}_{\hat{W}}(p'),$$

$$\forall p\in\Delta_{\mathcal{X}},\; P_e(p,\hat{W})=1-|\mathcal{X}|\int_{\Delta_{\mathcal{X}}}\max_{x\in\mathcal{X}}\left\{p(x)\times p'(x) \right\}\cdot d{\MP}_{\hat{W}}(p'),$$

$$\text{if}\;|\mathcal{X}|\geq 2,\; Z(\hat{W})=\frac{1}{|\mathcal{X}|-1}\sum_{\substack{x,x'\in\mathcal{X},\\x\neq x'}}\int_{\Delta_{\mathcal{X}}}\sqrt{p(x)p(x')}\cdot d{\MP}_{\hat{W}}(p),$$

$$\text{For every code}\;\mathcal{C}\subset\mathcal{X}^n,\;P_{e,\mathcal{C}}(\hat{W})= 1-\frac{|\mathcal{X}|^n}{|\mathcal{C}|}\int_{\Delta_{\mathcal{X}}^n}\max_{x_1^n\in\mathcal{C}}\left\{\prod_{i=1}^n p_i(x_i)\right\} d{\MP}_{\hat{W}}^n(p_1^n),$$
where $H(p)$ is the entropy of $p$, and ${\MP}_{\hat{W}}^n$ is the product measure on $\Delta_{\mathcal{X}}^n$ obtained by multiplying ${\MP}_{\hat{W}}$ with itself $n$ times. Note that we adopt the standard convention that $0\log\frac{0}{0}=0$.
\end{myprop}
\begin{proof}
By choosing any representative channel $W\in\hat{W}$ and replacing $W(y|x)$ by $|\mathcal{X}|P_W^o(y)W_y^{-1}(x)$ in the definitions of the channel parameters, all the above formulas immediately follow. Let us show how this works for $P_e$:
\begin{align*}
P_e(p,\hat{W})&=P_e(p,W)\stackrel{(a)}{=}1-\sum_{y\in\Imag(W)}\max_{x\in\mathcal{X}}\{p(x)W(y|x)\}\\
&=1-\sum_{y\in\Imag(W)}\max_{x\in\mathcal{X}}\big\{p(x)\cdot|\mathcal{X}|\cdot P_W^o(y)W_{y}^{-1}(x)\big\}\\
&=1-|\mathcal{X}|\sum_{y\in\Imag(W)}\max_{x\in\mathcal{X}}\{p(x)W_{y}^{-1}(x)\}\cdot P_W^o(y)\\
&=1-|\mathcal{X}|\int_{\Delta_{\mathcal{X}}}\max_{x\in\mathcal{X}}\{p(x)p'(x)\}\cdot d{\MP}_W(p')\\
&=1-|\mathcal{X}|\int_{\Delta_{\mathcal{X}}}\max_{x\in\mathcal{X}}\{p(x)p'(x)\}\cdot d{\MP}_{\hat{W}}(p'),
\end{align*}
where (a) is true because $W(y|x)=0$ for $y\notin\Imag(W)$.
\end{proof}

\begin{myprop}
\label{propContParamDMCXo}
Let $\mathcal{U}_{\mathcal{X}}$ be the standard topology on $\Delta_{\mathcal{X}}$. We have:
\begin{itemize}
\item $I:\Delta_{\mathcal{X}}\times\DMC_{\mathcal{X},\ast}^{(o)}\rightarrow \mathbb{R}^+$ is continuous on $(\Delta_{\mathcal{X}}\times\DMC_{\mathcal{X},\ast}^{(o)}, \mathcal{U}_{\mathcal{X}}\otimes\mathcal{T}_{\mathcal{X},\ast}^{(o)})$ and concave in $p$.
\item $C:\DMC_{\mathcal{X},\ast}^{(o)}\rightarrow \mathbb{R}^+$ is continuous on $(\DMC_{\mathcal{X},\ast}^{(o)},\mathcal{T}_{\mathcal{X},\ast}^{(o)})$.
\item $P_e:\Delta_{\mathcal{X}}\times\DMC_{\mathcal{X},\ast}^{(o)}\rightarrow [0,1]$ is continuous on $(\Delta_{\mathcal{X}}\times\DMC_{\mathcal{X},\ast}^{(o)}, \mathcal{U}_{\mathcal{X}}\otimes\mathcal{T}_{\mathcal{X},\ast}^{(o)})$ and concave in $p$.
\item $Z:\DMC_{\mathcal{X},\ast}^{(o)}\rightarrow [0,1]$ is continuous on $(\DMC_{\mathcal{X},\ast}^{(o)},\mathcal{T}_{\mathcal{X},\ast}^{(o)})$.
\item For every code $\mathcal{C}$ on $\mathcal{X}$, $P_{e,\mathcal{C}}:\DMC_{\mathcal{X},\ast}^{(o)}\rightarrow [0,1]$ is continuous on $(\DMC_{\mathcal{X},\ast}^{(o)},\mathcal{T}_{\mathcal{X},\ast}^{(o)})$.
\item For every $n>0$ and every $1\leq M\leq |\mathcal{X}|^n$, the mapping $P_{e,n,M}:\DMC_{\mathcal{X},\ast}^{(o)}\rightarrow [0,1]$ is continuous on $(\DMC_{\mathcal{X},\ast}^{(o)},\mathcal{T}_{\mathcal{X},\ast}^{(o)})$.
\end{itemize}
\end{myprop}
\begin{proof}
We associate the space $\mathcal{MP}(\mathcal{X})$ with the weak-$\ast$ topology. Define the mapping $$\overline{I}:\Delta_{\mathcal{X}}\times\mathcal{MP}(\mathcal{X})\rightarrow\mathbb{R}^+$$ as follows:
$$\overline{I}(p,{\MP})=H(p)-|\mathcal{X}|\cdot\int_{\Delta_{\mathcal{X}}}\left(\sum_{x\in\mathcal{X}} p(x)p'(x)\log\frac{p(x)p'(x)}{\displaystyle\sum_{x'}p(x')p'(x')} \right)\cdot d{\MP}(p'),$$
Lemma \ref{lemContProdWeakStar} implies that $\overline{I}$ is continuous. On the other hand, Proposition \ref{propFormulasParamMetaProb} shows that $I(p,\hat{W})=\overline{I}(p,{\MP}_{\hat{W}})$. Therefore, $I$ is continuous on $(\Delta_{\mathcal{X}}\times\DMC_{\mathcal{X},\ast}^{(o)},\mathcal{U}_{\mathcal{X}}\otimes\mathcal{T}_{\mathcal{X},\ast}^{(o)})$. We can prove the continuity of $P_e$ and $Z$ similarly.

Now define the mapping $\overline{C}:\mathcal{MP}(\mathcal{X})\rightarrow\mathbb{R}$ as
$$\overline{C}({\MP})=\sup_{p\in\Delta_{\mathcal{X}}}\overline{I}(p,{\MP}).$$
Fix ${\MP}\in\mathcal{MP}(\mathcal{X})$ and let $\epsilon>0$. Since $\mathcal{MP}(\mathcal{X})$ is compact (under the weak-$\ast$ topology), Lemma \ref{lemCompactProdCont} implies the existence of a weakly-$\ast$ open neighborhood $U_{{\MP}}$ of ${\MP}$ such that $|\overline{I}(p,{\MP})-\overline{I}(p,{\MP}')|<\epsilon$ for every ${\MP}'\in U_{{\MP}}$ and every $p\in\Delta_{\mathcal{X}}$. Therefore, for every ${\MP}'\in U_{{\MP}}$ and every $p\in\Delta_{\mathcal{X}}$, we have
$$\overline{I}(p,{\MP})< \overline{I}(p,{\MP}')+\epsilon\leq \overline{C}({\MP}')+\epsilon,$$
hence,
$$\overline{C}({\MP})=\sup_{p\in\Delta_{\mathcal{X}}}\overline{I}(p,{\MP})\leq \overline{C}({\MP}')+\epsilon.$$
Similarly, we can show that $\overline{C}({\MP}')\leq  \overline{C}({\MP})+\epsilon$. This shows that $|\overline{C}({\MP}')-\overline{C}({\MP})|\leq \epsilon$ for every ${\MP}'\in U_{{\MP}}$. Therefore, $\overline{C}$ is continuous. But $C(\hat{W})=\overline{C}({\MP}_{\hat{W}})$, so $C$ is continuous on $(\DMC_{\mathcal{X},\ast}^{(o)},\mathcal{T}_{\mathcal{X},\ast}^{(o)})$.

Now for every $0\leq i\leq n$, define the mapping $f_i:\Delta_{\mathcal{X}}^i\times \mathcal{MP}(\mathcal{X})\rightarrow\mathbb{R}$ backward-recursively as follows:
\begin{itemize}
\item $f_n(p_1^n,{\MP})=\displaystyle\max_{x_1^n\in\mathcal{C}}\left\{\prod_{i=1}^n p_i(x_i)\right\}$.
\item For every $0\leq i<n$, define
$$f_i(p_1^i,{\MP})=\int_{\Delta_{\mathcal{X}}} f_n(p_1^{i+1},{\MP})\cdot d{\MP}(p_{i+1}).$$
\end{itemize}

Clearly $f_n$ is continuous. Now let $0\leq i<n$ and assume that $f_{i+1}$ is continuous. If we let $S=\Delta_{\mathcal{X}}^i\times\mathcal{MP}(\mathcal{X})$, Lemma \ref{lemContProdWeakStar} implies that the mapping $F_i:\Delta_{\mathcal{X}}^i\times\mathcal{MP}(\mathcal{X})\times \mathcal{MP}(\mathcal{X})$ defined as
$$F_i(p_1^i,{\MP},{\MP}')=\int_{\Delta_{\mathcal{X}}} f(p_1^{i+1},{\MP})\cdot d{\MP}'(p_{i+1})$$
is continuous. But $f_i(p_1^i,{\MP})=F_i(p_1^i,{\MP},{\MP})$, so $f_i$ is also continuous. Therefore, $f_0$ is continuous. By noticing that $P_{e,\mathcal{C}}(\hat{W})=\displaystyle 1-\frac{|\mathcal{X}|^n}{|\mathcal{C}|}f_0({\MP}_{\hat{W}})$, we conclude that $P_{e,\mathcal{C}}$ is continuous on $(\DMC_{\mathcal{X},\ast}^{(o)},\mathcal{T}_{\mathcal{X},\ast}^{(o)})$. Moreover, since $P_{e,n,M}$ is the minimum of a finite family of continuous mappings, it is continuous.
\end{proof}

\vspace*{3mm}

It is worth mentioning that Proposition \ref{propContParamDMCXoStr} can be shown from Proposition \ref{propContParamDMCXo} because the noisiness topology is coarser than the strong topology.

\begin{mycor}
\label{corContParamDMCXoTV}
All the mappings in Proposition \ref{propContParamDMCXo} are also continuous if we replace the noisiness topology $\mathcal{T}_{\mathcal{X},\ast}^{(o)}$ with the total variation topology $\mathcal{T}_{TV,\mathcal{X},\ast}^{(o)}$.
\end{mycor}
\begin{proof}
This is true because $\mathcal{T}_{TV,\mathcal{X},\ast}^{(o)}$ is finer than $\mathcal{T}_{\mathcal{X},\ast}^{(o)}$.
\end{proof}

\subsection{Channel operations}

In the following, we show that we can express the channel operations in terms of Blackwell measures. We have all the tools to achieve this for the channel sum, channel product and channel interpolation. In order to express the channel polarization transformations in terms of the Blackwell measures, we need to introduce new definitions.

Let $\mathcal{X}$ be a finite set and let $\ast$ be a binary operation on a finite set $\mathcal{X}$. We say that $\ast$ is \emph{uniformity preserving} if the mapping $(a,b)\rightarrow (a\ast b,b)$ is a bijection from $\mathcal{X}^2$ to itself \cite{RajErgI}. For every $a,b\in\mathcal{X}$, we denote the unique element $c\in\mathcal{X}$ satisfying $c*b=a$ as $c=a/^{\ast} b$. Note that $/^{\ast}$ is a binary operation and it is uniformity preserving. $/^{\ast}$ is called the \emph{right-inverse} of $\ast$. It was shown in \cite{RajErgII} that a binary operation is polarizing if and only if it is uniformity preserving and its inverse is strongly ergodic.

Binary operations that are not uniformity preserving are not interesting for polarization theory because they do not preserve the symmetric capacity \cite{RajErgII}. Therefore, we will only focus on polarization transformations that are based on uniformity preserving operations.

Let $\ast$ be a fixed uniformity preserving operation on $\mathcal{X}$. Define the mapping $C^{-,\ast}:\Delta_{\mathcal{X}}\times \Delta_{\mathcal{X}}\rightarrow\Delta_{\mathcal{X}}$ as $$(C^{-,\ast}(p_1,p_2))(u_1)=\sum_{u_2\in\mathcal{X}}p_1(u_1\ast u_2)p_2(u_2).$$
The probability distribution $C^{-,\ast}(p_1,p_2)$ can be interpreted as follows: let $X_1$ and $X_2$ be two independent random variables in $\mathcal{X}$ that are distributed as $p_1$ and $p_2$ respectively, and let $(U_1,U_2)$ be the random pair in $\mathcal{X}^2$ defined as $(U_1,U_2)=(X_1/^{\ast}X_2,X_2)$, or equivalently $(X_1,X_2)=(U_1\ast U_2,U_2)$. $C^{-,\ast}(p_1,p_2)$ is the probability distribution of $U_1$.

Clearly, $C^{-,\ast}$ is continuous. Therefore, the push-forward mapping $C^{-,\ast}_{\#}$ is continuous from $\mathcal{P}(\Delta_{\mathcal{X}}\times \Delta_{\mathcal{X}})$ to $\mathcal{P}(\Delta_{\mathcal{X}})=\mathcal{MP}(\mathcal{X})$ under both the weak-$\ast$ and the total variation topologies (see Section \ref{subsecMetaProbDef}). For every $\MP_1,\MP_2\in\mathcal{MP}(\mathcal{X})$, we define the $(-,\ast)$-convolution of $\MP_1$ and $\MP_2$ as:
$$({\MP}_1,{\MP}_2)^{-,\ast}= C^{-,\ast}_{\#}({\MP}_1\times{\MP}_2)\in\mathcal{MP}(\mathcal{X}).$$
Since the product of meta-probability measures is continuous under both the weak-$\ast$ and the total variation topologies (Appendices \ref{appProdMeasureCont} and \ref{appContMetaProbProd}), the $(-,\ast)$-convolution is also continuous under these topologies.

For every $p_1,p_2\in\Delta_{\mathcal{X}}$ and every $u_1\in\supp(C^{-,\ast}(p_1,p_2))$, define $C^{+,u_1,\ast}(p_1,p_2)\in\Delta_{\mathcal{X}}$ as $$(C^{+,u_1,\ast}(p_1,p_2))(u_2)=\frac{p_1(u_1\ast u_2)p_2(u_2)}{(C^{-,\ast}(p_1,p_2))(u_1)}.$$
The probability distribution $C^{+,u_1,\ast}(p_1,p_2)$ can be interpreted as follows: if $X_1,X_2, U_1$ and $U_2$ are as above, $C^{+,u_1,\ast}(p_1,p_2)$ is the conditional probability distribution of $U_2$ given $U_1=u_1$.

Define the mapping $C^{+,\ast}:\Delta_{\mathcal{X}}\times \Delta_{\mathcal{X}}\rightarrow\mathcal{P}(\Delta_{\mathcal{X}})=\mathcal{MP}(\mathcal{X})$ as follows:
$$C^{+,\ast}(p_1,p_2)=\sum_{u_1\in\supp(C^{-,\ast}(p_1,p_2))} (C^{-,\ast}(p_1,p_2))(u_1)\cdot\delta_{C^{+,u_1,\ast}(p_1,p_2)},$$
where $\delta_{C^{+,u_1,\ast}(p_1,p_2)}$ is a Dirac measure centered at $C^{+,u_1,\ast}(p_1,p_2)$.

If $X_1,X_2,U_1$ and $U_2$ are as above, $C^{+,\ast}(p_1,p_2)$ is the meta-probability measure that describes the possible conditional probability distributions of $U_2$ that are seen by someone having knowledge of $U_1$. Clearly, $C^{+,\ast}$ is a random mapping from $\Delta_{\mathcal{X}}\times\Delta_{\mathcal{X}}$ to $\Delta_{\mathcal{X}}$. In Appendix \ref{appMeasContCplusstar}, we show that $C^{+,\ast}$ is a measurable random mapping. We also show in Appendix \ref{appMeasContCplusstar} that $C^{+,\ast}$ is a continuous mapping from $\Delta_{\mathcal{X}}\times\Delta_{\mathcal{X}}$ to $\mathcal{MP}(\mathcal{X})$ when the latter space is endowed with the weak-$\ast$ topology. Lemmas \ref{lemContPushForwardRandTV} and \ref{lemContPushForwardRandWeakStar} now imply that the push-forward mapping $C^{+,\ast}_{\#}$ is continuous under both the weak-$\ast$ and the total variation topologies.

For every $\MP_1,\MP_2\in\mathcal{MP}(\mathcal{X})$, we define the $(+,\ast)$-convolution of $\MP_1$ and $\MP_2$ as:
$$({\MP}_1,{\MP}_2)^{+,\ast}= C^{+,\ast}_{\#}({\MP}_1\times{\MP}_2)\in\mathcal{MP}(\mathcal{X}).$$
Since the product of meta-probability measures is continuous under both the weak-$\ast$ and the total variation topologies (Appendices \ref{appProdMeasureCont} and \ref{appContMetaProbProd}), the $(+,\ast)$-convolution is also continuous under these topologies.

\begin{myprop}
\label{propFormulasChanOperMetaProb}
We have:
\begin{itemize}
\item For every $\hat{W}_1\in\DMC_{\mathcal{X}_1,\ast}^{(o)}$ and $\overline{W}_2\in\DMC_{\mathcal{X}_2,\ast}^{(o)}$, we have:
$${\MP}_{\hat{W}_1\oplus \overline{W}_2}=\frac{|\mathcal{X}_1|}{|\mathcal{X}_1|+|\mathcal{X}_2|}{\MP}'_{\hat{W}_1}+\frac{|\mathcal{X}_2|}{|\mathcal{X}_1|+|\mathcal{X}_2|}{\MP}'_{\overline{W}_2},$$
where ${\MP}'_{\hat{W}_1}$ (respectively ${\MP}'_{\hat{W}_2}$) is the meta-push-forward of ${\MP}_{\hat{W}_1}$ (respectively ${\MP}_{\hat{W}_2}$) by the canonical injection from $\mathcal{X}_1$ (respectively $\mathcal{X}_2$) to $\mathcal{X}_1\coprod\mathcal{X}_2$.
\item For every $\hat{W}_1\in\DMC_{\mathcal{X}_1,\ast}^{(o)}$ and $\overline{W}_2\in\DMC_{\mathcal{X}_2,\ast}^{(o)}$, we have:
$${\MP}_{\hat{W}_1\otimes \overline{W}_2}={\MP}_{\hat{W}_1}\otimes {\MP}_{\overline{W}_2}.$$
\item For every $\alpha\in[0,1]$ and every $\hat{W}_1,\hat{W}_2\in\DMC_{\mathcal{X},\ast}^{(o)}$, we have
$${\MP}_{[\alpha\hat{W}_1,(1-\alpha)\hat{W}_2]}=\alpha{\MP}_{\hat{W}_1}+(1-\alpha){\MP}_{\hat{W}_2}.$$
\item For every uniformity preserving binary operation $\ast$ on $\mathcal{X}$, and every $\hat{W}\in\DMC_{\mathcal{X},\ast}^{(o)}$, we have
$${\MP}_{\hat{W}^-}=({\MP}_{\hat{W}},{\MP}_{\hat{W}})^{-,\ast}.$$
\item For every uniformity preserving binary operation $\ast$ on $\mathcal{X}$, and every $\hat{W}\in\DMC_{\mathcal{X},\ast}^{(o)}$, we have
$${\MP}_{\hat{W}^+}=({\MP}_{\hat{W}},{\MP}_{\hat{W}})^{+,\ast}.$$
\end{itemize}
\end{myprop}
\begin{proof}
See Appendix \ref{appFormulasChanOperMetaProb}.
\end{proof}

\vspace*{3mm}

Note that the polarization transformation formulas in Proposition \ref{propFormulasChanOperMetaProb} generalize the formulas given by Raginsky in \cite{RaginskyBlackwellPolar} for binary-input channels.

\begin{myprop}
\label{propContOperDMCXo}
Assume that all equivalent channel spaces are endowed with the noisiness/weak-$\ast$ or the total variation topology. We have:
\begin{itemize}
\item The mapping $(\hat{W}_1,\overline{W}_2)\rightarrow \hat{W}_1\oplus \overline{W}_2$ from $\DMC_{\mathcal{X}_1,\ast}^{(o)}\times \DMC_{\mathcal{X}_2,\ast}^{(o)}$ to $\DMC_{\mathcal{X}_1\coprod\mathcal{X}_2,\ast}^{(o)}$ is continuous.
\item The mapping $(\hat{W}_1,\overline{W}_2)\rightarrow \hat{W}_1\otimes \overline{W}_2$ from $\DMC_{\mathcal{X}_1,\ast}^{(o)}\times \DMC_{\mathcal{X}_2,\ast}^{(o)}$ to $\DMC_{\mathcal{X}_1\times\mathcal{X}_2,\ast}^{(o)}$ is continuous.
\item The mapping $(\hat{W}_1,\overline{W}_2,\alpha)\rightarrow [\alpha \hat{W}_1,(1-\alpha)\overline{W}_2]$ from $\DMC_{\mathcal{X},\ast}\times \DMC_{\mathcal{X},\ast}^{(o)}\times[0,1]$ to $\DMC_{\mathcal{X},\ast}^{(o)}$ is continuous.
\item For every uniformity preserving binary operation $\ast$ on $\mathcal{X}$, the mapping $\hat{W}\rightarrow \hat{W}^-$ from $\DMC_{\mathcal{X},\ast}^{(o)}$ to $\DMC_{\mathcal{X},\ast}^{(o)}$ is continuous.
\item For every uniformity preserving binary operation $\ast$ on $\mathcal{X}$, the mapping $\hat{W}\rightarrow \hat{W}^+$ from $\DMC_{\mathcal{X},\ast}^{(o)}$ to $\DMC_{\mathcal{X},\ast}^{(o)}$ is continuous.
\end{itemize}
\end{myprop}
\begin{proof}
The proposition directly follows from Proposition \ref{propFormulasChanOperMetaProb} and the fact that all the meta-probability measure operations that are involved in the formulas are continuous under both the weak-$\ast$ and the total variation topologies.
\end{proof}

\begin{mycor}
Both $(\DMC_{\mathcal{X},\ast}^{(o)},\mathcal{T}_{\mathcal{X},\ast}^{(o)})$ and $(\DMC_{\mathcal{X},\ast}^{(o)},\mathcal{T}_{TV,\mathcal{X},\ast}^{(o)})$ are strongly contractible to every point in $\DMC_{\mathcal{X},\ast}^{(o)}$.
\end{mycor}
\begin{proof}
We can use the same proof of Corollary \ref{corStrongContracStrong}.
\end{proof}

\section{Conclusion}

Sections \ref{secContDMCXYo} and \ref{secContMappingsStrong} show that the quotient topology is relatively easy to work with. If one is interested in the space of equivalent channels sharing the same input and output alphabets, then using the quotient formulation of the topology seems to be the easiest way to prove theorems.

The continuity of the channel sum and the channel product on the whole product space $(\DMC_{\mathcal{X}_1,\ast}^{(o)}\times \DMC_{\mathcal{X}_2,\ast}^{(o)},\mathcal{T}_{s,\mathcal{X}_1,\ast}^{(o)}\otimes \mathcal{T}_{s,\mathcal{X}_2,\ast}^{(o)})$ remains an open problem. As we mentioned in Section \ref{secContMappingsStrong}, it is sufficient to prove that the product topology $\mathcal{T}_{s,\mathcal{X}_1,\ast}^{(o)}\otimes \mathcal{T}_{s,\mathcal{X}_2,\ast}^{(o)}$ is compactly generated.

\section*{Acknowledgment}

I would like to thank Emre Telatar and Mohammad Bazzi for helpful discussions. I am also grateful to Maxim Raginsky for his comments.

\appendices

\section{Proof of Lemma \ref{lemCompactProdCont}}
\label{appCompactProdCont}

Fix $\epsilon>0$ and let $(s,t)\in S\times T$. Since $f$ is continuous, there exists a neighborhood $O_{s,t}$ of $(s,t)$ in $S\times T$ such that for every $(s',t')\in O_{s,t}$, we have $|f(s',t')-f(s,t)|<\frac{\epsilon}{2}$. Moreover, since products of open sets form a base for the product topology, there exists an open neighborhood $V_{s,t}$ of $s$ in $(S,\mathcal{V})$ and an open neighborhood $U_{s,t}$ of $t$ in $ T$ such that $V_{s,t}\times U_{s,t}\subset O_{s,t}$.

Since $(S,\mathcal{V})$ and $(T,\mathcal{U})$ are compact, the product space is also compact. On the other hand, we have $\displaystyle\bigcup_{(s,t)\in S\times T} V_{s,t}\times U_{s,t}=S\times T$ so $\{V_{s,t}\times U_{s,t}\}_{(s,t)\in S\times T}$ is an open cover of $S\times T$. Therefore, there exist $s_1,\ldots,s_n\in S$ and $t_1,\ldots,t_n\in T$ such that $\displaystyle\bigcup_{i=1}^n V_{s_i,t_i}\times U_{s_i,t_i}=S\times T$.

Now fix $s\in S$ and define $\displaystyle V_s=\bigcap_{\substack{1\leq i\leq n,\\ s\in V_{s_i,t_i}}} V_{s_i,t_i}$. Since $V_s$ is the intersection of finitely many open sets containing $s$, $V_s$ is an open neighborhood of $s$ in $(S,\mathcal{V})$. Let $s'\in V_s$ and $t\in T$. Since $\displaystyle\bigcup_{i=1}^n V_{s_i,t_i}\times U_{s_i,t_i}=S\times T$, there exists $1\leq i\leq n$ such that $(s,t)\in V_{s_i,t_i}\times U_{s_i,t_i}\subset O_{s_i,t_i}$. Since $s\in V_{s_i,t_i}$, we have $V_s\subset V_{s_i,t_i}$ and so $s'\in V_{s_i,t_i}$. Therefore, $(s',t)\in V_{s_i,t_i}\times U_{s_i,t_i}\subset O_{s_i,t_i}$, hence
$$|f(s',t)-f(s,t)|\leq|f(s',t)-f(s_i,t_i)|+|f(s_i,t_i)-f(s,t)|<\frac{\epsilon}{2}+\frac{\epsilon}{2}=\epsilon.$$
But this is true for every $t\in  T$. Therefore,
$$\sup_{t\in T}|f(s',t)-f(s,t)|\leq \epsilon.$$

\section{Continuity of the product of measures}

\label{appProdMeasureCont}

For every subset $A$ of $M_1\times M_2$ and every $x_1\in M_1$, define $A_2^{x_1}=\{x_2\in M_2:\;(x_1,x_2)\in A\}$. Similarly, for every $x_2\in M_2$, define $A_1^{x_2}=\{x_1\in M_1:\;(x_1,x_2)\in A\}$. Let $P_1,P_1'\in\mathcal{P}(M_1,\Sigma_1)$ and $P_2,P_2'\in\mathcal{P}(M_2,\Sigma_2)$. We have:
\begin{align*}
\|P_1\times P_2-& P_1'\times P_2'\|_{TV}=\sup_{A\in\Sigma_1\otimes\Sigma_2}|( P_1\times P_2)(A)-( P_1'\times P_2')(A)|\\
&\leq\sup_{A\in\Sigma_1\otimes\Sigma_2}\Big\{\big|( P_1\times P_2)(A)- ( P_1'\times P_2)(A)\big|+\big|( P_1'\times P_2)(A)- ( P_1'\times P_2')(A)\big|\Big\}\\
&
\begin{aligned}
=\sup_{A\in\Sigma_1\otimes\Sigma_2}&\Bigg\{\left|\int_{M_2}  P_1(A_1^{x_2})\cdot d P_2(x_2)- \int_{M_2}  P_1'(A_1^{x_2})\cdot d P_2(x_2) \right|\\
&\;\;\;\;\;\;\;\;\;\;\;\;\;\;\;\;\;\;\;\;\;\;\;\;\;\;+\left|\int_{M_1}  P_2(A_2^{x_1})\cdot d P_1'(x_1)- \int_{M_1}  P_2'(A_2^{x_1})\cdot d P_1'(x_1)\right|\Bigg\}
\end{aligned}\\
&\leq\sup_{A\in\Sigma_1\otimes\Sigma_2}\Bigg\{\int_{M_2} \left| P_1(A_1^{x_2})- P_1'(A_1^{x_2})\right|\cdot d P_2(x_2)+\int_{M_1} \left| P_2(A_2^{x_1})- P_2'(A_2^{x_1})\right|\cdot d P_1'(x_1)\Bigg\}\\
&\leq\int_{M_2} \left(\sup_{A_1\in\Sigma_1}\left| P_1(A_1)- P_1'(A_1)\right|\right)d P_2+\int_{M_1} \left(\sup_{A_2\in\Sigma_2} \left| P_2(A_2)- P_2'(A_2)\right|\right)d P_1'\\
&=\| P_1- P_1'\|_{TV}+\| P_2- P_2'\|_{TV}.
\end{align*}
This shows that the product of measures is continuous under the total variation topology.

\section{Proof of Proposition \ref{propPushForwardRandFormula}}
\label{appPushForwardRandFormula}
Define the mapping $G:M\rightarrow\mathbb{R}^+\cup\{+\infty\}$ as follows:
$$G(x)=\int_{M'} g(y) d(R(x))(y).$$

For every $n\geq 0$, define the mapping $g_n:M'\rightarrow\mathbb{R}^+$ as follows:
$$g_n(y)=\frac{1}{2^n}\big\lfloor 2^n\times \min\{n, g(y)\}\big\rfloor.$$
Clearly, for every $y\in M'$ we have:
\begin{itemize}
\item $g_n(y)\leq g(y)$ for all $n\geq 0$.
\item $g_n(y)\leq g_{n+1}(y)$ for all $n\geq 0$.
\item $\displaystyle\lim_{n\to\infty}g_n(y)=g(y)$.
\end{itemize}
Moreover, for every fixed $n\geq 0$, we have:
\begin{itemize}
\item $g_n$ is $\Sigma'$-measurable.
\item $g_n$ takes values in $\left\{\frac{i}{2^n}:\; 0\leq i\leq n2^n\right\}$.
\end{itemize}
For every $0\leq i\leq n2^n$, let $B_{i,n}=\{y\in M':\; g_n(y)=\frac{i}{2^n}\}$. Since $g_n$ is $\Sigma'$-measurable, we have $B_{i,n}\in\Sigma'$ for every $0\leq i\leq n2^n$. Now for every $n\geq 0$, define the mapping $G_n:M\rightarrow\mathbb{R}\cup\{+\infty\}$ as follows:
\begin{align*}
G_n(x)&=\int_{M'} g_n(y) d(R(x))(y)=\int_{M'}\left( \sum_{i=0}^{n2^n}\frac{i}{2^n}  \mathds{1}_{B_{i,n}}(y)\right)d(R(x))(y)\\
&=\sum_{i=0}^{n2^n}\frac{i}{2^n} (R(x))(B_{i,n})=\sum_{i=0}^{n2^n}\frac{i}{2^n}  R_{B_{i,n}}(x).
\end{align*}
Since the random mapping $R$ is measurable and since $B_{i,n}\in\Sigma'$, the mapping $R_{B_{i,n}}$ is $\Sigma$-measurable for every $0\leq i\leq n2^n$. Therefore, $G_n$ is $\Sigma$-measurable for every $n\geq 0$. Moreover, for every $x\in \Sigma$, we have:
\begin{align*}
\lim_{n\to\infty} G_n(x)&=\lim_{n\to\infty} \int_{M'} g_n(y) d(R(x))(y)\stackrel{(a)}{=}\int_{M'} g(y)d(R(x))(y)=G(x),
\end{align*}
where (a) follows from the monotone convergence theorem. We conclude that $G$ is $\Sigma$-measurable because it is the point-wise limit of $\Sigma$-measurable functions. On the other hand, we have
\begin{align*}
\int_{M'} g_n\cdot d(R_{\#}P)&=\sum_{i=0}^{n2^n}\frac{i}{2^n}(R_{\#}P)(B_{i,n})=\sum_{i=0}^{n2^n}\frac{i}{2^n} \int_{M} R_{B_{i,n}}(x)\cdot dP(x)\\
&=\sum_{i=0}^{n2^n}\frac{i}{2^n} \int_{M} (R(x))(B_{i,n})\cdot dP(x)
= \sum_{i=0}^{n2^n}\frac{i}{2^n}  \int_{M} \left(\int_{M'} \mathds{1}_{B_{i,n}}(y)\cdot d(R(x))(y)\right) dP(x)\\
&= \int_{M} \left(\int_{M'} \left(\sum_{i=0}^{n2^n}\frac{i}{2^n}  \mathds{1}_{B_{i,n}}(y)\right) d(R(x))(y)\right) dP(x)\\
&= \int_{M} \left(\int_{M'} g_n(y) d(R(x))(y)\right) dP(x)=\int_{M} G_n \cdot dP.
\end{align*}
Therefore,
$$\int_{M'} g\cdot d(R_{\#}P)\stackrel{(a)}{=} \lim_{n\to\infty} \int_{M'} g_n\cdot d(R_{\#}P)=\lim_{n\to\infty}\int_M G_n\cdot dP\stackrel{(b)}{=}\int_M G\cdot dP,$$
where (a) and (b) follow from the monotone convergence theorem.

\section{Continuity of the push-forward by a random mapping}
\label{appContPushForwardRand}

Let $R$ be a measurable random mapping from $(M,\Sigma)$ to $(M',\Sigma')$. Let $P_1,P_2\in\mathcal{P}(M,\Sigma)$. Define the signed measure $\mu=P_1-P_2$ and let $\{\mu^+,\mu^-\}$ be the Jordan measure decomposition of $\mu$. It is easy to see that $\|P_1-P_2\|_{TV}=\mu^+(M)=\mu^-(M)$. For every $B\in\Sigma'$, we have:
\begin{align*}
(R_{\#}(P_1))(B)-(R_{\#}(P_2))(B)&=\int_{M}R_B\cdot dP_1-\int_{M}R_B\cdot dP_2=\int_{M}R_B\cdot d(P_1-P_2)\\
&=\int_{M}R_B\cdot d(\mu^+-\mu^-)\leq \int_{M}R_B\cdot d\mu^+ \leq \|R_B\|_{\infty}\cdot \mu^+(M)\\
&\stackrel{(a)}{\leq}\mu^+(M)= \|P_1-P_2\|_{TV},
\end{align*}
where (a) follows from the fact that $|R_B(x)|=|(R(x))(B)|\leq 1$ for every $x\in M$. We can similarly show that
$$(R_{\#}(P_2))(B)-(R_{\#}(P_1))(B)\leq \|R_B\|_{\infty}\cdot \mu^-(M)\leq \|P_1-P_2\|_{TV}.$$
Therefore,
\begin{align*}
\|R_{\#}(P_1)-R_{\#}(P_2)\|_{TV}=\sup_{B\in\Sigma'}|(R_{\#}(P_1))(B)-(R_{\#}(P_2))(B)|\leq\|P_1-P_2\|_{TV}.
\end{align*}
This shows that the push-forward mapping $R_{\#}$ from $\mathcal{P}(M,\Sigma)$ to $\mathcal{P}(M',\Sigma')$ is continuous under the total variation topology. This concludes the proof of Lemma \ref{lemContPushForwardRandTV}.

Now assume that $\mathcal{U}$ is a Polish topology on $M$ and $\mathcal{U}'$ is an arbitrary topology on $M'$. Let $R$ be measurable random mapping from $(M,\mathcal{B}(M))$ to $(M',\mathcal{B}(M'))$. Moreover, assume that $R$ is a continuous mapping from $(M,\mathcal{U})$ to $\mathcal{P}(M',\mathcal{B}(M'))$ when the latter space is endowed with the weak-$\ast$ topology. Let $(P_n)_{n\geq 0}$ be a sequence of probability measures in $\mathcal{P}(M,\mathcal{B}(M))$ that weakly-$\ast$ converges to $P\in\mathcal{P}(M,\mathcal{B}(M))$.

Let $g:M'\rightarrow\mathbb{R}$ be a bounded and continuous mapping. Define the mapping $G:M\rightarrow \mathbb{R}$ as follows:
$$G(x)=\int_{M'} g(y)\cdot d(R(x))(y).$$
For every sequence $(x_n)_{n\geq 0}$ converging to $x$ in $M$, the sequence $(R(x_n))_{n\geq 0}$ weakly-$\ast$ converges to $R(x)$ in $\mathcal{P}(M',\mathcal{B}(M'))$ because of the continuity of $R$. This implies that the sequence $(G(x_n))_{n\geq 0}$ converges to $G(x)$. Since $\mathcal{U}$ is a Polish topology (hence metrizable and sequential \cite{SequentialSpace}), this shows that $G$ is a bounded and continuous mapping from $(M,\mathcal{U})$ to $\mathbb{R}$. Therefore, we have:
\begin{align*}
\lim_{n\to\infty} \int_{M'} g\cdot d(R_{\#}P_n)\stackrel{(a)}{=}\lim_{n\to\infty}\int_M G\cdot dP_n\stackrel{(b)}{=}\int_M G\cdot dP\stackrel{(c)}{=}\int_{M'} g\cdot d(R_{\#}P),
\end{align*}
where (a) and (c) follow from Corollary \ref{corPushForwardRandFormula}, and (b) follows from the fact that $(P_n)_{n\geq 0}$ weakly-$\ast$ converges to $P$. This shows that $(R_{\#}P_n)_{n\geq 0}$ weakly-$\ast$ converges to $R_{\#}P$. Now since $\mathcal{U}$ is Polish, the weak-$\ast$ topology on $\mathcal{P}(M,\mathcal{B}(M))$ is metrizable \cite{WassersteinMetric}, hence it is sequential \cite{SequentialSpace}. This shows that the push-forward mapping $R_{\#}$ from $\mathcal{P}(M,\mathcal{B}(M))$ to $\mathcal{P}(M',\mathcal{B}(M'))$ is continuous under the weak-$\ast$ topology.

\section{Proof of Lemma \ref{lemContProdWeakStar}}
\label{appContProdWeakStar}

For every $s\in S$, define the mapping $f_s:\Delta_{\mathcal{X}}\rightarrow\mathbb{R}$ as $f_s(p)=f(s,p)$. Clearly $f_s$ is continuous for every $s\in S$. Therefore, the mapping $F_s:\mathcal{MP}(\mathcal{X})\rightarrow\mathbb{R}$ defined as
$$F_s({\MP})=\int_{\Delta_{\mathcal{X}}} f_s\cdot d{\MP}$$
is continuous in the weak-$\ast$ topology of $\mathcal{MP}(\mathcal{X})$.

Fix $\epsilon>0$ and let $(s,{\MP})\in S\times\mathcal{MP}(\mathcal{X})$. Since $F_s$ is continuous, there exists a weakly-$\ast$ open neighborhood $U_{s,{\MP}}$ of ${\MP}$ such that $\displaystyle|F_s({\MP}')-F_s({\MP})|<\frac{\epsilon}{2}$ for every ${\MP}'\in U_{s,{\MP}}$. On the other hand, Lemma \ref{lemCompactProdCont} implies the existence of an open neighborhood $V_s$ of $s$ in $(S,\mathcal{V})$ such that for every $s'\in V_s$ we have
$$\sup_{p\in\Delta_{\mathcal{X}}}|f(s',p)-f(s,p)|\leq \frac{\epsilon}{2}.$$

Clearly $V_s\times U_{s,{\MP}}$ is an open neighborhood of $(s,{\MP})$ in $S\times\mathcal{MP}(\mathcal{X})$. For every $(s',{\MP}')\in V_s\times U_{s,{\MP}}$, we have
\begin{align*}
|F(s',{\MP}')-F(s,{\MP})|&\leq |F(s',{\MP}')-F(s,{\MP}')| + |F(s,{\MP}')-F(s,{\MP})|\\
&=\left|\int_{\Delta_{\mathcal{X}}}\big(f(s',p)-f(s,p)\big)\cdot d{\MP}'(p)\right|+|F_s({\MP}')-F_s({\MP})|\\
&<\left(\int_{\Delta_{\mathcal{X}}}|f(s',p)-f(s,p)|\cdot d{\MP}'(p)\right) + \frac{\epsilon}{2}\stackrel{(a)}{\leq} \frac{\epsilon}{2} + \frac{\epsilon}{2}=\epsilon,
\end{align*}
where (a) follows from the fact that ${\MP}'$ is a meta-probability measure and $|f(s',p)-f(s',p)|\leq\displaystyle\frac{\epsilon}{2}$ for every $p\in\Delta_{\mathcal{X}}$. We conclude that $F$ is continuous.

\section{Weak-$\ast$ continuity of the product of meta-probability measures}
\label{appContMetaProbProd}

Let $({\MP}_{1,n})_{n\geq 0}$ and $({\MP}_{2,n})_{n\geq 0}$ be two sequences that weakly-$\ast$ converge to ${\MP}_1$ and ${\MP}_2$ in $\mathcal{MP}(\mathcal{X}_1)$ and $\mathcal{MP}(\mathcal{X}_2)$ respectively. Let $f:\Delta_{\mathcal{X}_1}\times\Delta_{\mathcal{X}_2}\rightarrow\mathbb{R}$ be a continuous and bounded mapping. Define the mapping $F:\Delta_{\mathcal{X}_1}\times\mathcal{MP}(\mathcal{X}_2)$ as follows:
$$F(p_1,{\MP}_2')=\int_{\Delta_{\mathcal{X}_2}}f(p_1,p_2)d{\MP}'_2(p_2).$$

Fix $\epsilon>0$. Since $f(p_1,p_2)$ is continuous, Lemma \ref{lemContProdWeakStar} implies that $F$ is continuous. Therefore, the mapping $p_1\rightarrow F(p_1,\MP_2)$ is continuous on $\Delta_{\mathcal{X}_1}$, which implies that it is also bounded because $\Delta_{\mathcal{X}_1}$ is compact. Therefore,
$$\lim_{n\to\infty}\int_{\Delta_{\mathcal{X}_1}} F(p_1,{\MP}_2)d{\MP}_{1,n}(p_1)=\int_{\Delta_{\mathcal{X}_1}} F(p_1,{\MP}_2)d{\MP}_1(p_1)$$
because $({\MP}_{1,n})_{n\geq 0}$ weakly-$\ast$ converges to ${\MP}_1$. This means that there exists $n_1\geq 0$ such that for every $n\geq n_1$, we have
$$\left|\int_{\Delta_{\mathcal{X}_1}} F(p_1,{\MP}_2)d{\MP}_{1,n}(p_1)-\int_{\Delta_{\mathcal{X}_1}} F(p_1,{\MP}_2)d{\MP}_1(p_1)\right|<\frac{\epsilon}{2}.$$

On the other hand, since $F$ is continuous and since $\mathcal{MP}(\mathcal{X}_2)$ is compact under the weak-$\ast$ topology \cite{WassersteinMetric}, Lemma \ref{lemCompactProdCont} implies the existence of a weakly-$\ast$ open neighborhood $U_{{\MP}_2}$ of ${\MP}_2$ such that
$\displaystyle |F(p_1,{\MP}_2')-F(p_1,{\MP}_2)|\leq \frac{\epsilon}{2}$ for every ${\MP}_2'\in U_{{\MP}_2}$ and every $p_1\in\Delta_{\mathcal{X}_1}$. Moreover, since ${\MP}_{2,n}$ weakly-$\ast$ converges to ${\MP}_2$, there exists $n_2\geq 0$ such that ${\MP}_{2,n}\in U_{{\MP}_2}$ for every $n\geq n_2$.

Therefore, for every $n\geq \max\{n_1,n_2\}$, we have
\begin{align*}
&\left|\int_{\Delta_{\mathcal{X}_1}}\left(\int_{\Delta_{\mathcal{X}_2}} f(p_1,p_2)d{\MP}_{2,n}(p_2)\right)d{\MP}_{1,n}(p_1)-\int_{\Delta_{\mathcal{X}_1}}\left( \int_{\Delta_{\mathcal{X}_2}}f(p_1,p_2)d{\MP}_2(p_2)\right)d{\MP}_1(p_1)\right|\\
&\leq \left|\int_{\Delta_{\mathcal{X}_1}}\left(\int_{\Delta_{\mathcal{X}_2}} f(p_1,p_2)d{\MP}_{2,n}(p_2)\right)d{\MP}_{1,n}(p_1)-\int_{\Delta_{\mathcal{X}_1}}\left( \int_{\Delta_{\mathcal{X}_2}}f(p_1,p_2)d{\MP}_2(p_2)\right)d{\MP}_{1,n}(p_1)\right|\\
&\;\;+\left|\int_{\Delta_{\mathcal{X}_1}}\left(\int_{\Delta_{\mathcal{X}_2}} f(p_1,p_2)d{\MP}_2(p_2)\right)d{\MP}_{1,n}(p_1)-\int_{\Delta_{\mathcal{X}_1}}\left( \int_{\Delta_{\mathcal{X}_2}}f(p_1,p_2)d{\MP}_2(p_2)\right)d{\MP}_1(p_1)\right|\\
&=\left|\int_{\Delta_{\mathcal{X}_1}}\left(F(p_1,{\MP}_{2,n})-F(p_1,{\MP}_2)\right)d{\MP}_{1,n}(p_1)\right|\\
&\;\;\;\;\;\;\;\;\;\;\;\;\;\;\;\;\;\;\;\;\;\;\;\;\;\;\;\;\;\;\;\;\;\;\;\;\;\;\;\;\;\;\;\;\;\;\;\;\;\;\;\;\;\;\;\;+\left|\int_{\Delta_{\mathcal{X}_1}} F(p_1,{\MP}_2)d{\MP}_{1,n}(p_1)-\int_{\Delta_{\mathcal{X}_1}} F(p_1,{\MP}_2)d{\MP}_1(p_1)\right|\\
&<\int_{\Delta_{\mathcal{X}_1}}\left|F(p_1,{\MP}_{2,n})-F(p_1,{\MP}_2)\right|d{\MP}_{1,n}(p_1) + \frac{\epsilon}{2}\stackrel{(a)}{\leq}\int_{\Delta_{\mathcal{X}_1}}\frac{\epsilon}{2}\cdot d{\MP}_{1,n}(p_1)+\frac{\epsilon}{2}= \epsilon,
\end{align*}
where (a) follows from the fact ${\MP}_{2,n}\in U_{{\MP}_2}$ for every $n\geq n_2$. Therefore,
\begin{align*}
\lim_{n\to\infty} \int_{\Delta_{\mathcal{X}_1}\times \Delta_{\mathcal{X}_2}} f\cdot d({\MP}_{1,n}\times {\MP}_{2,n})&\stackrel{(a)}{=}\lim_{n\to\infty} \int_{\Delta_{\mathcal{X}_1}}\left(\int_{\Delta_{\mathcal{X}_2}} f(p_1,p_2)d{\MP}_{2,n}(p_2)\right)d{\MP}_{1,n}(p_1)\\
&=\int_{\Delta_{\mathcal{X}_1}}\left( \int_{\Delta_{\mathcal{X}_2}}f(p_1,p_2)d{\MP}_2(p_2)\right)d{\MP}_1(p_1)\\
&\stackrel{(b)}{=}\int_{\Delta_{\mathcal{X}_1}\times \Delta_{\mathcal{X}_2}} f\cdot d({\MP}_1\times {\MP}_2),
\end{align*}
where (a) and (b) follow from Fubini's theorem. We conclude that $({\MP}_{1,n}\times{\MP}_{2,n})_{n\geq 0}$ weakly-$\ast$ converges to $({\MP}_1\times{\MP}_2)_{n\geq 0}$. Therefore the product of meta-probability measures is weakly-$\ast$ continuous.

\section{Continuity of the capacity}
\label{appContParamDMCXY}

Since the mapping $I$ is continuous, and since the space $\Delta_{\mathcal{X}}\times\DMC_{\mathcal{X},\mathcal{Y}}$ is compact, the mapping $I$ is uniformly continuous, i.e., for every $\epsilon>0$, there exists $\delta(\epsilon)>0$ such that for every $(p_1,W_1), (p_2,W_2)\in \Delta_{\mathcal{X}}\times\DMC_{\mathcal{X},\mathcal{Y}}$, if $\|p_1-p_2\|_1:=\displaystyle\sum_{x\in\mathcal{X}}|p_1(x)-p_2(x)|<\delta(\epsilon)$ and $d_{\mathcal{X},\mathcal{Y}}(W_1,W_2)<\delta(\epsilon)$, then $$|I(p_1,W_1)-I(p_2,W_2)|<\epsilon.$$

Let $W_1,W_2\in\DMC_{\mathcal{X},\mathcal{Y}}$ be such that $d_{\mathcal{X},\mathcal{Y}}(W_1,W_2)<\delta(\epsilon)$. For every $p\in\Delta_{\mathcal{X}}$, we  have $\|p-p\|_1=0<\delta(\epsilon)$ so we must have $|I(p,W_1)-I(p,W_2)|<\epsilon$. Therefore,
\begin{align*}
I(p,W_1)< I(p,W_2) + \epsilon\leq \sup_{p'\in\Delta_{\mathcal{X}}} I(p',W_2) + \epsilon = C(W_2) + \epsilon.
\end{align*}
Therefore,
\begin{align*}
C(W_1)= \sup_{p\in\Delta_{\mathcal{X}}} I(p,W_1)\leq C(W_2) + \epsilon.
\end{align*}
Similarly, we can show that $C(W_2)\leq C(W_1)+\epsilon$. This implies that $|C(W_1)-C(W_2)|\leq\epsilon$, hence $C$ is continuous.

\section{Measurability and continuity of $C^{+,\ast}$}
\label{appMeasContCplusstar}

Let us first show that the random mapping $C^{+,\ast}$ is measurable. We need to show that the mapping $C^{+,\ast}_B:\Delta_{\mathcal{X}}\times\Delta_{\mathcal{X}}\rightarrow\mathbb{R}$ is measurable for every $B\in\mathcal{B}(\Delta_{\mathcal{X}})$, where
$$C^{+,\ast}_B(p_1,p_2)=(C^{+,\ast}(p_1,p_2))(B),\;\;\forall p_1,p_2\in\Delta_{\mathcal{X}}.$$

For every $u_1\in\mathcal{X}$, define the set $$A_{u_1}=\{(p_1,p_2)\in\Delta_{\mathcal{X}}\times\Delta_{\mathcal{X}}:\;(C^{-,\ast}(p_1,p_2))(u_1)>0\}.$$
Clearly, $A_{u_1}$ is open in $\Delta_{\mathcal{X}}\times\Delta_{\mathcal{X}}$ (and so it is measurable). The mapping $C^{+,u_1,\ast}$ is defined on $A_{u_1}$ and it is clearly continuous. Therefore, for every $B\in\mathcal{B}(\Delta_{\mathcal{X}})$, $(C^{+,u_1,\ast})^{-1}(B)$ is measurable. We have:
\begin{align*}
C^{+,\ast}_B(p_1,p_2)&=(C^{+,\ast}(p_1,p_2))(B)=\sum_{\substack{u_1\in\supp(C^{-,\ast}(p_1,p_2)),\\C^{+,u_1,\ast}(p_1,p_2)\in B}} (C^{-,\ast}(p_1,p_2))(u_1)\\
&=\sum_{\substack{u_1\in \mathcal{X},\\(p_1,p_2)\in A_{u_1},\\C^{+,u_1,\ast}(p_1,p_2)\in B}} (C^{-,\ast}(p_1,p_2))(u_1)\stackrel{(a)}{=}\sum_{u_1\in \mathcal{X}} (C^{-,\ast}(p_1,p_2))(u_1)\cdot\mathds{1}_{(C^{+,u_1,\ast})^{-1}(B)}(p_1,p_2),
\end{align*}
where (a) follows from the fact that $(p_1,p_2)\in (C^{+,u_1,\ast})^{-1}(B)$ if and only if $(p_1,p_2)\in A_{u_1}$ and $C^{+,u_1,\ast}(p_1,p_2)\in B$. This shows that $C^{+,\ast}_B$ is measurable for every $B\in\mathcal{B}(\Delta_{\mathcal{X}})$. Therefore, $C^{+,\ast}$ is a measurable random mapping.

Let $(p_{1,n},p_{2,n})_{n\geq 0}$ be a converging sequence to $(p_1,p_2)$ in $\Delta_{\mathcal{X}}\times\Delta_{\mathcal{X}}$. Since $C^{-,\ast}$ is continuous, we have $\displaystyle \lim_{n\to\infty} (C^{-,\ast}(p_{1,n},p_{2,n}))(u_1)=(C^{-,\ast}(p_1,p_2))(u_1)$ for every $u_1\in\mathcal{X}$. Therefore, for every $u_1\in \supp(C^{-,\ast}(p_1,p_2))$, there exists $n_{u_1}\geq 0$ such that for every $n\geq n_{u_1}$, we have $C^{-,\ast}(p_{1,n},p_{2,n})>0$. Let $n_0=\max\{n_{u_1}:\; u_1\in\supp(C^{-,\ast}(p_1,p_2))\}$. For every $n\geq n_0$, we have $\supp(C^{-,\ast}(p_1,p_2))\subset\supp(C^{-,\ast}(p_{1,n},p_{2,n}))$. Therefore, for every continuous and bounded mapping $g:\Delta_{\mathcal{X}}\rightarrow \mathbb{R}$, we have
\begin{align*}
\lim_{n\to\infty} \int_{\Delta_{\mathcal{X}}} g\cdot d(C^{+,\ast}(p_{1,n},p_{2,n}))&=\lim_{n\to\infty} \sum_{u_1\in \supp(C^{-,\ast}(p_{1,n},p_{2,n}))} g(C^{+,u_1,\ast}(p_{1,n},p_{2,n}))\cdot (C^{-,\ast}(p_{1,n},p_{2,n}))(u_1)\\
&\stackrel{(a)}{=}\lim_{n\to\infty} \sum_{u_1\in \supp(C^{-,\ast}(p_1,p_2))} g(C^{+,u_1,\ast}(p_{1,n},p_{2,n}))\cdot (C^{-,\ast}(p_{1,n},p_{2,n}))(u_1)\\
&\stackrel{(b)}{=} \sum_{u_1\in \supp(C^{-,\ast}(p_1,p_2))} g(C^{+,u_1,\ast}(p_1,p_2))\cdot (C^{-,\ast}(p_1,p_2))(u_1)\\
&= \int_{\Delta_{\mathcal{X}}} g\cdot d(C^{+,\ast}(p_1,p_2)),
\end{align*}
where (b) follows from the continuity of $g$ and $C^{-,\ast}$, and the continuity of $C^{+,u_1,\ast}$ on $A_{u_1}$ for every $u_1\in\mathcal{X}$. (a) follows from the fact that:
\begin{align*}
\lim_{n\to\infty} \sum_{\substack{ u_1\in \supp(C^{-,\ast}(p_{1,n},p_{2,n})),\\ u_1\notin \supp(C^{-,\ast}(p_1,p_2))}} &\big| g(C^{+,u_1,\ast}(p_{1,n},p_{2,n}))\cdot (C^{-,\ast}(p_{1,n},p_{2,n}))(u_1)\big|\\
&\leq  \|g\|_{\infty} \lim_{n\to\infty}\sum_{\substack{ u_1\in \supp(C^{-,\ast}(p_{1,n},p_{2,n})),\\ u_1\notin \supp(C^{-,\ast}(p_1,p_2))}}  (C^{-,\ast}(p_{1,n},p_{2,n}))(u_1)\\
&= \|g\|_{\infty} \lim_{n\to\infty}\left(1- \sum_{u_1\in \supp(C^{-,\ast}(p_1,p_2))}  (C^{-,\ast}(p_{1,n},p_{2,n}))(u_1)\right)\\
&= \|g\|_{\infty}\left(1- \sum_{u_1\in \supp(C^{-,\ast}(p_1,p_2))}  (C^{-,\ast}(p_1,p_2))(u_1)\right)=0.
\end{align*}

We conclude that the mapping $C^{+,\ast}$ is a continuous mapping from $\Delta_{\mathcal{X}}\times\Delta_{\mathcal{X}}$ to $\mathcal{MP}(\mathcal{X})$ when the latter space is endowed with the weak-$\ast$ topology.

\section{Proof of Proposition \ref{propFormulasChanOperMetaProb}}
\label{appFormulasChanOperMetaProb}

Let $\hat{W}_1\in\DMC_{\mathcal{X}_1,\ast}^{(o)}$ and $\overline{W}_2\in\DMC_{\mathcal{X}_2,\ast}^{(o)}$. Fix $W_1\in\hat{W}_1$ and $W_2\in\overline{W}_2$ and let $\mathcal{Y}_1$ and $\mathcal{Y}_2$ be the output alphabets of $W_1$ and $W_2$ respectively. We may assume without loss of generality that $\Imag(W_1)=\mathcal{Y}_1$ and $\Imag(W_2)=\mathcal{Y}_2$.

Let $y\in\mathcal{Y}_1$. We have
\begin{align*}
P_{W_1\oplus W_2}^o(y)&=\frac{1}{|\mathcal{X}_1\coprod\mathcal{X}_2|}\sum_{x\in \mathcal{X}_1 \coprod\mathcal{X}_2} (W_1\oplus W_2)(y|x)\\
&=\frac{1}{|\mathcal{X}_1|+|\mathcal{X}_2|}\sum_{x\in \mathcal{X}_1} W_1(y|x)=\frac{|\mathcal{X}_1|}{|\mathcal{X}_1|+|\mathcal{X}_2|}P_{W_1}^o(y)>0.
\end{align*}
For every $x\in\mathcal{X}_1$, we have
$$(W_1\oplus W_2)^{-1}_y(x)=\frac{(W_1\oplus W_2)(y|x)}{(|\mathcal{X}_1|+|\mathcal{X}_2|)P_{W_1}^o(y)}=\frac{W_1(y|x)}{|\mathcal{X}_1| P_{W_1}^o(y)}=(W_1)_y^{-1}(x).$$
On the other hand, for every $x\in\mathcal{X}_2$, we have
$$(W_1\oplus W_2)^{-1}_y(x)=\frac{(W_1\oplus W_2)(y|x)}{(|\mathcal{X}_1|+|\mathcal{X}_2|)P_{W_1}^o(y)}=0.$$
Therefore $(W_1\oplus W_2)_y^{-1} = \phi_{1\#}(W_1)_y^{-1}$, where $\phi_1$ is the canonical injection from $\mathcal{X}_1$ to $\mathcal{X}_1\coprod \mathcal{X}_2$.

Similarly, for every $y\in\mathcal{Y}_2$, we have $\displaystyle P_{W_1\oplus W_2}^o(y)=\frac{|\mathcal{X}_2|}{|\mathcal{X}_1|+|\mathcal{X}_2|}P_{W_1}^o(y)>0$ and $(W_1\oplus W_2)_y^{-1} = \phi_{2\#}(W_2)_y^{-1}$, where $\phi_2$ is the canonical injection from $\mathcal{X}_2$ to $\mathcal{X}_1\coprod \mathcal{X}_2$. For every $B\in\mathcal{B}(\Delta_{\mathcal{X}_1\coprod\mathcal{X}_2})$, we have:
\begin{align*}
{\MP}_{W_1\oplus W_2}(B)&=\sum_{\substack{y\in\mathcal{Y}_1\coprod\mathcal{Y}_2,\\(W_1\oplus W_2)^{-1}_y\in B}} P_{W_1\oplus W_2}^o(y)\\
&=\Bigg(\sum_{\substack{y\in\mathcal{Y}_1,\\\phi_{1\#}(W_1)^{-1}_y\in B}} \frac{|\mathcal{X}_1|}{|\mathcal{X}_1|+|\mathcal{X}_2|} P_{W_1}^o(y)\Bigg)+\Bigg(\sum_{\substack{y\in\mathcal{Y}_2,\\\phi_{2\#}(W_2)^{-1}_y\in B}} \frac{|\mathcal{X}_2|}{|\mathcal{X}_1|+|\mathcal{X}_2|} P_{W_2}^o(y)\Bigg)\\
&=\frac{|\mathcal{X}_1|}{|\mathcal{X}_1|+|\mathcal{X}_2|}{\MP}_{W_1}\big((\phi_{1\#})^{-1}(B)\big)+\frac{|\mathcal{X}_2|}{|\mathcal{X}_1|+|\mathcal{X}_2|}{\MP}_{W_2}\big((\phi_{2\#})^{-1}(B)\big)\\
&=\frac{|\mathcal{X}_1|}{|\mathcal{X}_1|+|\mathcal{X}_2|}(\phi_{1\#\#}{\MP}_{W_1})(B)+\frac{|\mathcal{X}_2|}{|\mathcal{X}_1|+|\mathcal{X}_2|}(\phi_{2\#\#}{\MP}_{W_2})(B).
\end{align*}
Therefore,
$${\MP}_{\hat{W}_1\oplus \overline{W}_2}=\frac{|\mathcal{X}_1|}{|\mathcal{X}_1|+|\mathcal{X}_2|}\phi_{1\#\#}{\MP}_{\hat{W}_1}+\frac{|\mathcal{X}_2|}{|\mathcal{X}_1|+|\mathcal{X}_2|}\phi_{2\#\#}{\MP}_{\overline{W}_2}.$$
This shows the first formula of Proposition \ref{propFormulasChanOperMetaProb}.

For every $y=(y_1,y_2)\in\mathcal{Y}_1\times\mathcal{Y}_2$, we have
\begin{align*}
P_{W_1\otimes W_2}^o(y)&=\sum_{(x_1,x_2)\in\mathcal{X}_1\times\mathcal{X}_2} \frac{1}{|\mathcal{X}_1\times\mathcal{X}_2|}(W_1\otimes W_2)(y_1,y_2|x_1,x_2)\\
&=\sum_{\substack{x_1\in\mathcal{X}_2,\\x_2\in\mathcal{X}_2}} \frac{W_1(y_1|x_1)}{|\mathcal{X}_1|}\cdot\frac{W_2(y_2|x_2)}{|\mathcal{X}_2|}=P_{W_1}^o(y_1)P_{W_2}^o(y_2)>0.
\end{align*}
For every $x=(x_1,x_2)\in\mathcal{X}_1\times\mathcal{X}_2$, we have
\begin{align*}
(W_1\otimes W_2)_y^{-1}(x)&=\frac{(W_1\otimes W_2)(y|x)}{|\mathcal{X}_1\times\mathcal{X}_2|P_{W_1\otimes W_2}^o(y)}=\frac{W_1(y_1|x_1)}{|\mathcal{X}_1|P_{W_1}^o(y_1)}\cdot \frac{W_2(y_2|x_2)}{|\mathcal{X}_2|P_{W_2}^o(y_2)}\\
&=(W_1)_{y_1}^{-1}(x_1)\cdot (W_2)_{y_2}^{-1}(x_2)=\big((W_1)_{y_1}^{-1}\times (W_2)_{y_2}^{-1}\big)(x).
\end{align*}
For every $B\in\mathcal{B}(\Delta_{\mathcal{X}_1\times\mathcal{X}_2})$, we have
\begin{align*}
{\MP}_{W_1\otimes W_2}(B)&=\sum_{\substack{y\in\mathcal{Y}_1\times\mathcal{Y}_2,\\(W_1\otimes W_2)^{-1}_y\in B}} P_{W_1\otimes W_2}^o(y)=\sum_{\substack{y\in\mathcal{Y}_1\times\mathcal{Y}_2,\\(W_1)_{y_1}^{-1}\times (W_2)_{y_2}^{-1}\in B}} P_{W_1}^o(y_1)P_{W_2}^o(y_2)\\
&=\sum_{\substack{y\in\mathcal{Y}_1\times\mathcal{Y}_2,\\\Mul\left((W_1)_{y_1}^{-1}, (W_2)_{y_2}^{-1}\right)\in B}} P_{W_1}^o(y_1)P_{W_2}^o(y_2)=({\MP}_{W_1}\times{\MP}_{W_2})({\Mul}^{-1}(B))\\
&=\big({\Mul}_{\#}({\MP}_{W_1}\times{\MP}_{W_2})\big)(B)=({\MP}_{W_1}\otimes{\MP}_{W_2})(B).
\end{align*}
Therefore,
$${\MP}_{\hat{W}_1\otimes \overline{W}_2}={\MP}_{\hat{W}_1}\otimes {\MP}_{\overline{W}_2}.$$
This shows the second formula of Proposition \ref{propFormulasChanOperMetaProb}.

Now let $\alpha\in[0,1]$ and $\hat{W}_1,\hat{W}_2\in\DMC_{\mathcal{X},\ast}^{(o)}$. Fix $W_1\in\hat{W}_1$ and $W_2\in\hat{W}_2$ and let $\mathcal{Y}_1$ and $\mathcal{Y}_2$ be the output alphabets of $W_1$ and $W_2$ respectively. We may assume without loss of generality that $\Imag(W_1)=\mathcal{Y}_1$ and $\Imag(W_2)=\mathcal{Y}_2$. Let $W=[\alpha W_1,(1-\alpha)W_2]$. If $\alpha=0$, then $W$ is equivalent to $W_2$ and $\MP_W=\MP_{W_2}=\alpha\MP_{W_1}+(1-\alpha)\MP_{W_2}$. If $\alpha=1$, then $W$ is equivalent to $W_1$ and $\MP_W=\MP_{W_1}=\alpha\MP_{W_1}+(1-\alpha)\MP_{W_2}$.

Assume now that $0<\alpha<1$. For every $y\in\mathcal{Y}_1$, we have:
\begin{align*}
P_{W}^o(y)=\frac{1}{|\mathcal{X}|}\sum_{x\in\mathcal{X}} W(y|x)=\frac{1}{|\mathcal{X}|}\sum_{x\in\mathcal{X}} \alpha\cdot W_1(y|x)=\alpha P_{W_1}^o(y) >0.
\end{align*}
For every $x\in\mathcal{X}$, we have:
\begin{align*}
W_y^{-1}(x)=\frac{W(y|x)}{|\mathcal{X}|P_W^o(y)}=\frac{\alpha W_1(y|x)}{|\mathcal{X}|\alpha P_{W_1}^o(y)}=(W_1)_y^{-1}(x).
\end{align*}

Similarly, for every $y\in \mathcal{Y}_2$, we have $P_{W}^o(y)=(1-\alpha) P_{W_2}^o(y) >0$ and $W_y^{-1}=(W_2)_y^{-1}$. Therefore,
\begin{align*}
{\MP}_W&=\sum_{y\in\mathcal{Y}_1\coprod\mathcal{Y}_2}P_W^o(y)\cdot\delta_{W_y^{-1}}=\left(\sum_{y\in\mathcal{Y}_1}\alpha P_{W_1}^o(y)\cdot\delta_{(W_1)_y^{-1}}\right)+\left(\sum_{y\in\mathcal{Y}_2}(1-\alpha)P_{W_2}^o(y)\cdot\delta_{(W_2)_y^{-1}}\right)\\
&=\alpha{\MP}_{W_1} + (1-\alpha){\MP}_{W_2}.
\end{align*}
Therefore,
$${\MP}_{[\alpha\hat{W}_1,(1-\alpha)\hat{W}_2]}=\alpha{\MP}_{\hat{W}_1}+(1-\alpha){\MP}_{\hat{W}_2}.$$
This shows the third formula of Proposition \ref{propFormulasChanOperMetaProb}.

Now let $\hat{W}\in\DMC_{\mathcal{X},\ast}^{(o)}$ and let $\ast$ be a uniformity preserving binary operation on $\mathcal{X}$. Fix $W\in\hat{W}$ and let $\mathcal{Y}$ be the output alphabet of $W$. We may assume without loss of generality that $\Imag(W)=\mathcal{Y}$.

Let $U_1,U_2$ be two independent random variables uniformly distributed in $\mathcal{X}$. Let $X_1=U_1\ast U_2$ and $X_2=U_2$. Send $X_1$ and $X_2$ through two independent copies of $W$ and let $Y_1$ and $Y_2$ be the output respectively.

For every $(y_1,y_2)\in\mathcal{Y}^2$, we have
$$P_{W^-}^o(y_1,y_2)=P_{Y_1,Y_2}(y_1,y_2)=P_{Y_1}(y_1)P_{Y_2}(y_2)=P_W^o(y_1)P_W^o(y_2)>0.$$
For every $u_1\in\mathcal{X}$, we have:
\begin{align*}
(W^-)_{y_1,y_2}^{-1}(u_1)&=P_{U_1|Y_1,Y_2}(u_1|y_1,y_2)=\sum_{u_2\in\mathcal{X}_2}P_{U_1,U_2|Y_1,Y_2}(u_1,u_2|y_1,y_2)\\
&=\sum_{u_2\in\mathcal{X}_2}P_{X_1,X_2|Y_1,Y_2}(u_1\ast u_2,u_2|y_1,y_2)=\sum_{u_2\in\mathcal{X}_2}P_{X_1|Y_1}(u_1\ast u_2|y_1)P_{X_2|Y_2}(u_2|y_2)\\
&=\sum_{u_2\in\mathcal{X}_2}W_{y_1}^{-1}(u_1\ast u_2)W_{y_2}^{-1}(u_2)=\big(C^{-,\ast}(W_{y_1}^{-1},W_{y_2}^{-1})\big)(u_1).
\end{align*}

For every $B\in \mathcal{B}(\Delta_{\mathcal{X}})$, we have
\begin{align*}
{\MP}_{W^-}(B)&=\sum_{\substack{y\in\mathcal{Y}^2,\\(W^-)^{-1}_y\in B}} P_{W^-}^o(y)=\sum_{\substack{(y_1,y_2)\in\mathcal{Y}^2,\\C^{-,\ast}(W_{y_1}^{-1},W_{y_2}^{-1})\in B}} P_{W_1}^o(y_1)P_{W_2}^o(y_2)\\
&=({\MP}_W\times{\MP}_W)\big((C^{-,\ast})^{-1}(B)\big)=\big(C^{-,\ast}_{\#}({\MP}_W\times{\MP}_W)\big)(B)=({\MP}_W,{\MP}_W)^{-,\ast}(B).
\end{align*}
Therefore,
$${\MP}_{\hat{W}^-}=({\MP}_{\hat{W}},{\MP}_{\hat{W}})^{-,\ast}.$$
This shows the forth formula of Proposition \ref{propFormulasChanOperMetaProb}.

For every $(y_1,y_2,u_1)\in\mathcal{Y}^2\times\mathcal{X}$, we have:
\begin{align*}
P_{W^+}^o(y_1,y_2,u_1)&=P_{Y_1,Y_2,U_1}(y_1,y_2,y_1)=P_{Y_1,Y_2}(y_1,y_2)P_{U_1|Y_1,Y_2}(u_1|y_1,y_2)\\
&=P_{W}^o(y_1)P_{W}^o(y_2)\cdot\big(C^{-,\ast}(W_{y_1}^{-1},W_{y_2}^{-1})\big)(u_1).
\end{align*}
Therefore,
$$\Imag(W^+)=\bigcup_{(y_1,y_2)\in\mathcal{Y}^2}\{(y_1,y_2)\}\times\supp(C^{-,\ast}(W_{y_1}^{-1},W_{y_2}^{-1})).$$

For every $(y_1,y_2,u_1)\in\Imag(W^+)$, we have:
\begin{align*}
(W^+)^{-1}_{y_1,y_2,u_1}(u_2)&=P_{U_2|Y_1,Y_2,U_1}(u_2|y_1,y_2,u_1)=\frac{P_{U_1,U_2|Y_1,Y_2}(u_1,u_2|y_1,y_2)}{P_{U_1|Y_1,Y_2}(u_1|y_1,y_2)}\\
&=\frac{P_{X_1|Y_1}(u_1\ast u_2|y_1)P_{X_2|Y_2}(u_2|y_2)}{\big(C^{-,\ast}(W_{y_1}^{-1},W_{y_2}^{-1})\big)(u_1)}=\frac{W_{y_1}^{-1}(u_1\ast u_2)W_{y_2}^{-1}(u_2)}{\big(C^{-,\ast}(W_{y_1}^{-1},W_{y_2}^{-1})\big)(u_1)}\\
&=\big(C^{+,u_1,\ast}(W_{y_1}^{-1},W_{y_2}^{-1})\big)(u_2).
\end{align*}

For every $B\in\mathcal{B}(\Delta_{\mathcal{X}})$, we have
\begin{align*}
{\MP}_{W^+}(B)&=\sum_{(y_1,y_2)\in\mathcal{Y}^2}\sum_{\substack{u_1\in \supp(C^{-,\ast}(W_{y_1}^{-1},W_{y_2}^{-1}),\\C^{+,u_1,\ast}(W_{y_1}^{-1},W_{y_2}^{-1})\in B}} P_{W}^o(y_1)P_{W}^o(y_2)\cdot\big(C^{-,\ast}(W_{y_1}^{-1},W_{y_2}^{-1})\big)(u_1)\\
&=\sum_{(y_1,y_2)\in\mathcal{Y}^2} P_{W}^o(y_1)P_{W}^o(y_2) \sum_{\substack{u_1\in \supp(C^{-,\ast}(W_{y_1}^{-1},W_{y_2}^{-1}),\\C^{+,u_1,\ast}(W_{y_1}^{-1},W_{y_2}^{-1})\in B}} \big(C^{-,\ast}(W_{y_1}^{-1},W_{y_2}^{-1})\big)(u_1)\\
&=\sum_{(y_1,y_2)\in\mathcal{Y}^2} P_{W}^o(y_1)P_{W}^o(y_2)
\big(C^{+,\ast}(W_{y_1}^{-1},W_{y_2}^{-1})\big)(B)\\
&=\sum_{(y_1,y_2)\in\mathcal{Y}^2} P_{W}^o(y_1)P_{W}^o(y_2)
(C^{+,\ast}_B(W_{y_1}^{-1},W_{y_2}^{-1})\\
&=\int_{\Delta_{\mathcal{X}}\times \Delta_{\mathcal{X}}} C^{+,\ast}_B(p_1,p_2)\cdot d({\MP}_{W}\times {\MP}_{W})(p_1,p_2)\\
&=\big(C^{+,\ast}_{\#}({\MP}_{W}\times {\MP}_{W})\big)(B)=({\MP}_W,{\MP}_W)^{+,\ast}(B).
\end{align*}
Therefore,
$${\MP}_{\hat{W}^+}=({\MP}_{\hat{W}},{\MP}_{\hat{W}})^{+,\ast}.$$
This shows the fifth and last formula of Proposition \ref{propFormulasChanOperMetaProb}.

\bibliographystyle{IEEEtran}
\bibliography{bibliofile}
\end{document}